\documentclass[conference]{IEEEtran}
\IEEEoverridecommandlockouts

\usepackage[noadjust]{cite}
\usepackage{amsmath,amssymb,amsfonts}
\usepackage{algorithmic}
\usepackage{graphicx}
\usepackage{textcomp}
\usepackage{xcolor}
\def\BibTeX{{\rm B\kern-.05em{\sc i\kern-.025em b}\kern-.08em
    T\kern-.1667em\lower.7ex\hbox{E}\kern-.125emX}}

\usepackage{quiver}
\usepackage{amsthm}
\usepackage{xypic}
\usepackage{hhline}

\newtheorem{definition}{Definition}
\newtheorem{proposition}[definition]{Proposition}
\newtheorem{lemma}[definition]{Lemma}
\newtheorem{theorem}[definition]{Theorem}
\newtheorem{corollary}[definition]{Corollary}

\newcommand{\TV}[0]{TV}
\newcommand{\tlc}[2]{\begin{tabular}{@{}c@{}} #1 \vspace{0.2em} \\ #2\end{tabular}}
\newcommand{\thlc}[3]{\begin{tabular}{@{}c@{}} #1 \\ #2 \\ #3 \end{tabular}}

\newcommand{\hide}[1]{}

\begin{document}

\title{Substructural Abstract Syntax with Variable Binding and Single-Variable
  Substitution%
}

\author{\IEEEauthorblockN{Marcelo Fiore\\ marcelo.fiore@cl.cam.ac.uk}
\and
\IEEEauthorblockN{Sanjiv Ranchod\\ sanjiv.ranchod@cl.cam.ac.uk}
\\
\IEEEauthorblockA{\mbox{}\hspace*{-80mm}Department of Computer Science and
Technology, University of Cambridge} 
}

\maketitle

\begin{abstract}
We develop a unified categorical theory of substructural abstract syntax with
variable binding and single-variable (capture-avoiding) substitution.  
This is done for the gamut of context structural rules given by
exchange~(linear theory) with weakening~(affine theory) or with
contraction~(relevant theory) and with both~(cartesian theory).
Specifically, in all four scenarios, 
we uniformly: 
define abstract syntax with variable binding as free algebras for
binding-signature endofunctors over variables;
provide finitary algebraic axiomatisations of the laws of substitution;
construct single-variable substitution operations by generalised structural
recursion; and
prove their correctness, establishing their universal abstract character as
initial substitution algebras.
%
\end{abstract}

\medskip 


\section*{Introduction}

The algebraic study of languages with variable binding was initiated by Fiore,
Plotkin and Turi~\cite{FPT}, and by Gabbay and Pitts~\cite{GabbayPitts}.
Therein, abstract syntax for variable-binding operators was characterised by
means of initial-algebra semantics, 
thus equipping it with definitions by structural recursion and reasoning
principles by structural induction.

The theory of Fiore, Plotkin and Turi~\cite{FPT} considered a universe of
discourse given by (covariant) presheaves (or variable sets) on the category
of contexts.  
In the mono-sorted (or uni-typed, or untyped) case, contexts are generated
from one sort by iterated context extension.  Every context $\Gamma$ has an
associated set of variables (or indices) $V\Gamma$.  
The structure of these contexts admitted for all structural rules;
namely~exchange, weakening, and contraction.  Thus, context morphisms from a
context $\Delta$ to another one $\Gamma$ allowed for arbitrary variable
renamings; that is, are functions from $V\Delta$ to $V\Gamma$.
Mathematically, this determines the free cocartesian category on one object,
$\mathbb F$, equivalent to that of finite cardinals and functions between them.
The mathematical universe of discourse is then the presheaf category 
$\mathcal F = \mathbf{Set}^{\mathbb F}$.
Therein, Fiore, Plotkin and Turi~\cite{FPT} axiomatised single-variable as
well as simultaneous substitution and, for abstract syntax with variable
binding over binding signatures, derived provably-correct constructions for
both notions of substitution by structural recursion.
The importance of correctness is clear; that of structural recursion resides
in guaranteeing well-defined operations.  

Tanaka~\cite{Tanaka} considered part of the aforementioned \emph{cartesian}
theory 
in the \emph{linear} setting, where the structure of contexts only allows the
exchange structural rule. 
In this case, the category of contexts is the free symmetric strict monoidal
category on one object, $\mathbb B$, equivalent to that of finite cardinals
and bijections between them.  The mathematical universe of discourse is then
the presheaf category $\mathcal B = \mathbf{Set}^{\mathbb B}$ of Joyal's
combinatorial species of structures~\cite{J81,J86}.  
To note is that Tanaka~\cite{Tanaka} only revisited, and obtained analogous
results for, the theory of simultaneous substitution, leaving the development
of the theory of single-variable substitution as an open problem.  
A main motivation and contribution of this paper is to provide a solution to
it; but, moreover, to do so in a unified framework for substructural systems.

Subsequently, Tanaka and Power~\cite{TanakaPower} developed a general semantic
framework, showing it to encompass, among others, the theories of
cartesian~\cite{FPT}, linear~\cite{Tanaka}, and also affine (namely, with
exchange and weakening) abstract syntax with variable binding and, again,
simultaneous substitution.

The focus of this work is instead the theory of single-variable substitution
for substructural abstract syntax with variable binding that has so far been
neglected.
Notwithstanding, we contend that the development of both theories of
substitution, single-variable and simultaneous, is important.  
Indeed, not only in computer science and logic, but also in category theory
and algebra, both notions of substitution have been considered and studied in
depth: 
in category theory, in reference to the notion of 
multicategory; 
in algebra, in reference to the notion of operad. 
Moreover, none of the notions is more fundamental than the other: as it is
well-known, while simultaneous substitution is derived from single-variable
substitution by iteration, single-variable substitution may be specialised
from simultaneous substitution.  

{\scriptsize\begin{figure*}[t!]
\begin{center}\begin{tabular}{|c|c||c|c|c|c|} 
\hline
\textbf{\S} & \textbf{Topic}  & \textbf{Cartesian} & \textbf{Linear} &
\textbf{Affine} & \textbf{Relevant} 
\\ \hhline{|======|}
& \tlc{Category of}{contexts $\mathbb{C}$} & \tlc{$\mathbb{F}$}{Functions} &
\tlc{$\mathbb{B}$}{Bijections} & \tlc{$\mathbb{I}$}{Injections} &
\tlc{$\mathbb{S}$}{Surjections} 
\\ \cline{2-6}
\ref{Subsection:CategoriesOfContexts}  & \tlc{Monoidal}{Structure $\otimes$} &
Cocartesian & Symmetric & Semicocartesian & Corelevant 
\\ \cline{2-6}
& \tlc{Structure on}{objects} & \tlc{Symmetric}{monoid} &
\tlc{Symmetric}{object} & \tlc{Symmetric}{pointed object} &
\thlc{Symmetric}{multiplicative}{object} 
\\ \hline
& \tlc{Universe of}{discourse $\mathcal{C}$} &
$\mathcal{F} = \mathbf{Set}^\mathbb{F}$ & 
$\mathcal{B}= \mathbf{Set}^\mathbb{B}$ & 
$\mathcal{I} = \mathbf{Set}^\mathbb{I}$ &
$\mathcal{S} = \mathbf{Set}^\mathbb{S}$ 
\\ \cline{2-6}
\ref{subsec_UniDisc} 
& Day tensor $\hat\otimes$ & Cartesian & Symmetric & Semicartesian & Relevant 
\\ \cline{2-6}
& \tlc{Presheaf of}{variables $V$} & \tlc{Symmetric}{comonoid} &
\tlc{Symmetric}{object} & \tlc{Symmetric}{copointed object} &
\thlc{Symmetric}{comultiplicative}{object} 
\\ \hline
%
%
%
& \tlc{Structure on}{endofunctor $- \hat\otimes X$} & \tlc{Symmetric}{comonad}
& \tlc{Symmetric}{endofunctor} & \thlc{Symmetric}{copointed}{endofunctor} &
\thlc{Symmetric}{comultiplicative}{endofunctor} 
\\ \cline{2-6}
\raisebox{4mm}{\ref{ContextExtension}, \ref{SymmetricEndofunctors}}
& Context extension $\delta$ & \tlc{Symmetric}{monad} &
\tlc{Symmetric}{endofunctor} & \thlc{Symmetric}{pointed}{endofunctor} &
\thlc{Symmetric}{multiplicative}{endofunctor} 
\hide{
\\ \hline
\tlc{\ref{CartesianSubstitutionAlgebras}, \ref{subsec_linDerFun}}
{\ref{subsec_affDerFun}, \ref{RelevantDerivedFunctors}} & 
\tlc{Product Rule}{$\delta(X \hat\otimes Y)$} & 
$\delta(X) \hat\otimes \delta(Y)$  & 
$\delta(X) \hat\otimes Y + X \hat\otimes \delta(Y)$ &
\tlc{$\delta(X) \hat\otimes Y + X \hat\otimes \delta(Y)$}{$+ X \hat\otimes Y$}
& \tlc{$\delta(X) \hat\otimes Y + X \hat\otimes \delta(Y)$}
{$+ \delta(X) \hat\otimes \delta(Y)$} 
}
\\ \hline
\end{tabular}\end{center}
\caption{Summary of Sections~\ref{Section:CategoricalBackground} and 
  \ref{Section:StructuralEndofunctors}}
\label{CLARfigure}
\end{figure*}}
{\small\begin{figure*}[t]\begin{center}\begin{tabular}{|c|c||c|c|c|c|} 
\hline
\textbf{\S}  & \textbf{Topic}  & \textbf{Cartesian} & \textbf{Linear} &
\textbf{Affine} & \textbf{Relevant} 
\\ \hhline{|======|}
\tlc{\ref{CartesianSubstitutionAlgebras}, \ref{subsec_linDerFun}}
{\ref{subsec_affDerFun}, \ref{RelevantDerivedFunctors}} & 
\tlc{Product Rule}
    {for $\delta(X \hat\otimes Y)$} 
& \tlc{$\delta(X) \hat\otimes \delta(Y)$}{ \phantom{ }}
& \tlc{$\delta(X) \hat\otimes Y + X \hat\otimes \delta(Y)$} 
      {\small (Leibniz 
               Rule)}
& \tlc{$\delta(X) \hat\otimes Y + X \hat\otimes \delta(Y)$}
      {$+\ X \hat\otimes Y$}
& \tlc{$\delta(X) \hat\otimes Y + X \hat\otimes \delta(Y)$}
      {$+\ \delta(X) \hat\otimes \delta(Y)$} 
\\ \hline
\end{tabular}\end{center}
\caption{Product rules}
\label{Figure:ProductRules}
\end{figure*}}

The theory developed here considers the gamut of substructural
systems of the diagram
below~(\emph{c.f.}~\cite{NotesOnCombinatorialFunctors2001}). 
%
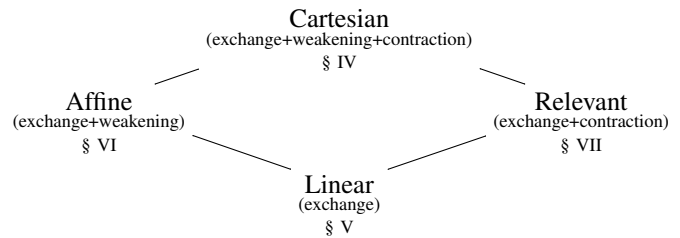
\[\xymatrix@R1pt@C2pt{
& \ar@{-}[dl] 
{\underset{\text{\S~\ref{Section:CartesianTheory}}}
  {\underset{\text{(exchange+weakening+contraction)}}{\text{Cartesian}}}}
\ar@{-}[dr] & 
\\
{\underset{\text{\S~\ref{Section:AffineTheory}}}
  {\underset{\text{(exchange+weakening)}}{\text{Affine}}}}
& & 
{\underset{\text{\S~\ref{Section:RelevantTheory}}}
  {\underset{\text{(exchange+contraction)}}{\text{Relevant}}}}
\\
& \ar@{-}[ul] 
{\underset{\text{\S~\ref{Section:LinearTheory}}}
  {\underset{\text{(exchange)}}{\text{Linear}}}}
\ar@{-}[ur] & 
}\]
For all of the above, our contributions are: 
\begin{itemize}
\item 
  to define abstract syntax with variable binding; 
\item 
  to algebraically axiomatise the laws of single-variable substitution by
  means of finitary equational presentations;
\item 
  to construct single-variable substitution operations by generalised
  structural recursion; and
\item 
  to prove their correctness and establish their universal abstract character
  as initial substitution algebras.
\end{itemize}

The cartesian theory is presented in Section~\ref{Section:CartesianTheory}.
There, we reconstruct the theory of Fiore, Plotkin and Turi~\cite{FPT} for
single-variable substitution, albeit using generalised structural recursion to
give a new construction of single-variable substitution
and new direct proofs of its correctness.  This has direct application to the
mathematical derivation of dependently-typed programs for single-variable
substitution.

The linear and affine theories are respectively presented in
Sections~\ref{Section:LinearTheory} and~\ref{Section:AffineTheory}.  
The axiomatisation of linear single-variable substitution is equivalent to the
standard notion of symmetric operad (or one-object symmetric multicategory),
while that of affine single-variable substitution is an expected extension.
This conceptually justifies our approach.

The relevant theory is developed in Section~\ref{Section:RelevantTheory}.  We
know of no previous work considering it.

We emphasise that our overall development is uniform throughout the gamut of
substructural theories.  This is achieved by an abstract analysis of the
mathematical structure of context structural
rules~(\emph{c.f.}~\cite{FioreMoggiSangiorgiLICS, FPT, MFPS2006, CT2007}). 
The resulting categorical theory is the topic of
Sections~\ref{Section:CategoricalBackground}
and~\ref{Section:StructuralEndofunctors}, and encompasses the mathematical
structures of Fig.~\ref{CLARfigure}.
Besides this, the approach crucially necessitates the consideration of
`product rules' as in Fig.~\ref{Figure:ProductRules}
from which theories of `derived functors'
(Sections~\ref{subsec_linDerFun},~\ref{subsec_affDerFun},
and~\ref{RelevantDerivedFunctors}) arise to address, and deal with, the
complexities and nuances of single-variable substitution in each scenario.

\begin{figure*}
\[\begin{tikzcd}[ampersand replacement=\&,row sep=scriptsize]
	{A^{\otimes 2}} \& {A^{\otimes 2}} \\
	\& {A^{\otimes 2}}
	\arrow["s", from=1-1, to=1-2]
	\arrow["{\mathbf{(a)}}"{description, pos=0.65}, shift left=3, draw=none, from=1-1, to=2-2]
	\arrow["s", from=1-2, to=2-2]
	\arrow["{{\mathrm{id}}}"', from=1-1, to=2-2]
\end{tikzcd} \quad
\begin{tikzcd}[ampersand replacement=\&,row sep=scriptsize]
	{A^{\otimes 3}} \& {A^{\otimes 3}} \& {A^{\otimes 3}} \\
	{A^{\otimes 3}} \& {A^{\otimes 3}} \& {A^{\otimes 3}}
	\arrow["{\mathrm{id}\otimes s}", from=1-1, to=1-2]
	\arrow["{s \otimes \mathrm{id}}"', from=1-1, to=2-1]
	\arrow["{s \otimes \mathrm{id}}", from=1-2, to=1-3]
	\arrow["{\mathrm{id}\otimes s}", from=1-3, to=2-3]
	\arrow["{\mathrm{id}\otimes s}"', from=2-1, to=2-2]
	\arrow["{s \otimes \mathrm{id}}"', from=2-2, to=2-3]
        \arrow["{\mathbf{(b)}}"{description}, draw=none, from=1-2, to=2-2]
\end{tikzcd}\quad
\begin{tikzcd}[ampersand replacement=\&,row sep=scriptsize]
	A \& {A^{\otimes 2}} \\
	\& {A^{\otimes 2}}
	\arrow["{\mathbf{(c)}}"{description, pos=0.65}, shift left=3, draw=none, from=1-1, to=2-2]
	\arrow["{w \otimes \mathrm{id}}", from=1-1, to=1-2]
	\arrow["{\mathrm{id} \otimes w}"', from=1-1, to=2-2]
	\arrow["s", from=1-2, to=2-2]
\end{tikzcd}\quad
\begin{tikzcd}[ampersand replacement=\&, row sep=scriptsize]
	A \& {A^{\otimes 2}} \\
	{A^{\otimes 2}} \& A
	\arrow["{\mathbf{(d)}}"{description, pos=0.65}, shift left=3, draw=none, from=1-1, to=2-2]
	\arrow["{\mathrm{id}\otimes w}", from=1-1, to=1-2]
	\arrow["{w \otimes \mathrm{id}}"', from=1-1, to=2-1]
	\arrow["{\mathrm{id}}"', from=1-1, to=2-2]
	\arrow["c", from=1-2, to=2-2]
	\arrow["c"', from=2-1, to=2-2]
\end{tikzcd}\]
\[\begin{tikzcd}[ampersand replacement=\&,cramped,row sep=scriptsize]
	{A^{\otimes 3}} \& {A^{\otimes 2}} \\
	{A^{\otimes 2}} \& A
	\arrow["{\mathbf{(e)}}"{description}, draw=none, from=1-1, to=2-2]
	\arrow["{{c \otimes \mathrm{id}}}", from=1-1, to=1-2]
	\arrow["{{\mathrm{id} \otimes c}}"', from=1-1, to=2-1]
	\arrow["c", from=1-2, to=2-2]
	\arrow["c"', from=2-1, to=2-2]
\end{tikzcd}\quad
\begin{tikzcd}[ampersand replacement=\&,cramped,row sep=scriptsize]
	{A^{\otimes 2}} \& {A^{\otimes 2}} \\
	\& A
	\arrow["{\mathbf{(f)}}"{description, pos=0.65}, shift left=3, draw=none, from=1-1, to=2-2]
	\arrow["s", from=1-1, to=1-2]
	\arrow["c"', from=1-1, to=2-2]
	\arrow["c", from=1-2, to=2-2]
\end{tikzcd}\quad
\begin{tikzcd}[ampersand replacement=\&,cramped,row sep=scriptsize]
	{A^{\otimes 3}} \& {A^{\otimes 3}} \& {A^{\otimes 3}} \\
	{A^{\otimes 2}} \&\& {A^{\otimes 2}}
	\arrow["{\mathbf{(g)}}"{description}, draw=none, from=1-1, to=2-3]
	\arrow["{{s \otimes \mathrm{id}}}", from=1-1, to=1-2]
	\arrow["{{\mathrm{id} \otimes c}}"', from=1-1, to=2-1]
	\arrow["{{\mathrm{id} \otimes s}}", from=1-2, to=1-3]
	\arrow["{{c \otimes \mathrm{id}}}", from=1-3, to=2-3]
	\arrow["s"', from=2-1, to=2-3]
\end{tikzcd}
\]
\caption{Symmetric object properties}
\label{fig_SymOpProp}
\end{figure*}

\section{Categorical Background}
\label{Section:CategoricalBackground}

\subsection{Categories of contexts}
\label{Subsection:CategoriesOfContexts}

The structural operations on contexts ---weakening, contraction, and
exchange--- may respectively be captured by structural morphisms, 
$w: I \to A$, $c: A \otimes A \to A$ and $s: A \otimes A \to A \otimes A$, on
an object $A$ in a monoidal category $(\otimes, I)$, subject to appropriate
axioms.  For a single-sorted theory, the category of contexts may then be
modelled as the PRO (PROduct category~\cite{MacLaneAlgebra,Markl}) on an
object with structural morphisms corresponding to the operations permitted by
the substructural theory.  The following definition respectively provides such
objects corresponding to linear, affine, relevant, and cartesian theories.

\begin{definition}
Let $A$ be an object in a monoidal category $(
\otimes, I)$.
\begin{enumerate}
\item 
  $(A, s)$ is a \emph{symmetric object} if it satisfies \textbf{(a)} and
  \textbf{(b)} in Fig.~\ref{fig_SymOpProp}.  

\item 
  $(A, w, s)$ is a \emph{symmetric pointed object} if it satisfies
  \textbf{(a)}--\textbf{(c)} in Fig.~\ref{fig_SymOpProp}.

\item 
  $(A, c, s)$ is a \emph{symmetric multiplicative object} if it satisfies 
  \textbf{(a)}, \textbf{(b)}, and \textbf{(e)}--\textbf{(g)} in 
  Fig.~\ref{fig_SymOpProp}.

\item (Grandis~\cite{G01}) 
  $(A, c, w, s)$ is a \emph{symmetric monoid} if it satisfies 
  \textbf{(a)}--\textbf{(g)} in Fig.~\ref{fig_SymOpProp}.
\end{enumerate}
\end{definition}

\begin{definition}
The \emph{categories of linear ($\mathbb{B}$), affine ($\mathbb{I}$),
relevant~($\mathbb{S}$), and cartesian ($\mathbb{F}$) contexts} are
respectively defined as the free strict monoidal categories on a symmetric
object, symmetric pointed object, symmetric multiplicative object and
symmetric monoid.
\end{definition}

The four categories of the above definition are equivalent to categories whose
objects are the sets $\mathbf{n} = \{1, \ldots, n\}$ for each natural number
$n$ with morphisms as follows.  In the case of $\mathbb{B}$ they are the
bijections, 
in the case of $\mathbb{I}$ they are the injections, 
in the case of $\mathbb{S}$ they are the surjections, 
and in the case of $\mathbb{F}$ they are all functions 
We adopt this presentation, writing the respective generating symmetric
object, symmetric pointed object, symmetric multiplicative object and
symmetric monoid as $(\mathbf{1}, s)$, $(\mathbf{1}, w, s)$,
$(\mathbf{1}, c, s)$, and $(\mathbf{1}, c, w, s)$.

Recall that a symmetric monoidal category 
is \emph{semicocartesian} (\textit{resp}.~\textit{semicartesian}) when the
monoidal unit is initial (\textit{resp}.~terminal) in the category; 
it is \textit{corelevant} (\textit{resp}.~\textit{relevant}) when it is
equipped with codiagonals (\textit{resp}.~diagonals); and 
it is \emph{cocartesian} (\textit{resp}.~\textit{cartesian}) when the monoidal
tensor is given by the categorical coproduct (\textit{resp}.~categorical
product).
We have the following alternative characterisations.

\begin{proposition}\label{prop_ctxSymMon}
\begin{enumerate}
\item 
  $\mathbb{B}$ is the free symmetric strict monoidal category on one object,
  $\mathbf{1}$.  

\item 
  $\mathbb{I}$ is the free semicocartesian strict monoidal category on one
  object, $\mathbf{1}$.

\item 
  $\mathbb{S}$ is the free corelevant strict monoidal category on one object,
  $\mathbf{1}$.  

\item 
  $\mathbb{F}$ is the free cocartesian strict monoidal category on one object,
  $\mathbf{1}$.
\end{enumerate}
\end{proposition}

We therefore, in the remainder of this section, take the 
\textit{category of contexts} as some (small) monoidal category 
$(\mathbb{C}, \otimes, I)$, while also considering specialisations of
structures in the cases that $\mathbb{C}$ is symmetric, semicocartesian,
corelevant, or cocartesian monoidal. We refer to objects of $\mathbb{C}$ as
\textit{contexts} and morphisms as \textit{context renamings}. The tensor
product models \textit{context concatenation}. 

\subsection{The universe of discourse}
\label{subsec_UniDisc}

The category in which we model syntax and substitution for a given theory is
the category of (covariant) presheaves over the category of contexts, 
$\mathcal{C} = \mathbf{Set}^\mathbb{C}$. 
An object $X$ in $\mathcal{C}$ is intuitively understood to consist of `terms'
of an abstract syntax $X$, organised into contexts, together with `term
renamings' induced by `context renamings'.

As a category of presheaves, $\mathcal{C}$ is a Grothendieck topos.  It is
thus complete and cocomplete, with limits and colimits given pointwise in
$\mathbf{Set}$.
It is equipped with a cartesian monoidal structure
$(\mathcal{C}, \times, 1)$ and a cocartesian monoidal structure 
$(\mathcal{C}, +, 0)$. 
The cartesian structure is closed, with the right adjoint to each 
$(-) \times X: \mathcal{C} \to \mathcal{C}$ written as $(-)^X$. 

$\mathcal{C}$ exhibits a further monoidal structure induced by the monoidal tensor on $\mathbb{C}$, referred to as the \textit{Day convolution} \cite{D70, IK86}. Writing $\mathcal{Y}: \mathbb{C}^\mathrm{op} \to \mathcal{C}$ for the Yoneda embedding on $\mathbb{C}$, we have the following definition.

\begin{definition}
The \textit{Day convolution}, 
$(-) \hat{\otimes} (-) : \mathcal{C} \times \mathcal{C} \to \mathcal{C}$ is
induced as the left Kan extension of 
$\mathcal{Y}(- \otimes -) 
 : \mathbb{C}^\mathrm{op} \times \mathbb{C}^\mathrm{op} \to \mathcal{C}$ 
along $\mathcal{Y} \times \mathcal{Y}$ and is given, for each 
$X, Y \in \mathcal{C}$ and $A \in \mathbb{C}$, by the coend formula:
\begin{equation*}\textstyle
(X \hat{\otimes} Y)(A) 
= \int^{B_1, B_2 \in \mathbb{C}} 
    X(B_1) \times Y(B_2) \times \mathbb{C}(B_1\!\otimes\!B_2, A) 
\end{equation*}
\end{definition}

The Day convolution has $J = \mathcal{Y}(I)$ as its monoidal unit and equips
$\mathcal{C}$ with a monoidal structure $(\mathcal{C}, \hat{\otimes}, J)$.
Furthermore, this monoidal structure is also closed, with the right adjoint to
$(-) \hat{\otimes} X$ (for each $X \in \mathcal{C}$) written as 
$X \multimap (-) : \mathcal{C} \to \mathcal{C}$, and given by the end formula:
\begin{equation*}\textstyle
(X \multimap Y)(A) 
= \int_{B \in \mathbb{C}} 
    \mathbf{Set}\big(\, X(B) \, , \, Y(A \otimes B) \, \big)
\end{equation*}

The Day convolution, being canonically induced by the monoidal tensor on the
category of contexts, models pairings of terms. Because both the Day and the
cartesian tensor are closed, they distribute over the cocartesian tensor. That
is, we have the (canonical) natural isomorphisms 
$X \times Z + Y \times Z \cong (X + Y) \times Z$ and 
$X \hat\otimes Z + Y \hat{\otimes} Z \cong (X + Y) \hat{\otimes} Z$. With
respect to the specific categories of contexts, we have the following.

\begin{proposition}\label{prop_daySymMon}
    \begin{enumerate}
        \item The Day convolution on $\mathcal{B} = \mathbf{Set}^\mathbb{B}$ induces a symmetric monoidal structure.
        \item The Day convolution on $\mathcal{I} = \mathbf{Set}^\mathbb{I}$ induces a semicartesian monoidal structure.
        \item The Day convolution on $\mathcal{S} = \mathbf{Set}^\mathbb{S}$ induces a relevant monoidal structure.
        \item The Day convolution on $\mathcal{F} = \mathbf{Set}^\mathbb{F}$ induces a cartesian monoidal structure.
    \end{enumerate}
\end{proposition}

In these concrete cases, the generating object, $\mathbf{1}$, is equipped with
a dual structure in the opposite category, which is preserved by the Yoneda
embedding. We refer to $V = \mathcal{Y}(\mathbf{1})$ as the 
\textit{presheaf of variables}.

\begin{proposition}\label{prop_VSymOb}
    \begin{enumerate}
        \item $\mathbf{1}$ and $V$ are symmetric objects in $\mathbb{B}^\mathrm{op}$ and $\mathcal{B}$, respectively.
        \item $\mathbf{1}$ and $V$ are symmetric copointed objects in $\mathbb{I}^\mathrm{op}$ and $\mathcal{I}$, respectively.
        \item $\mathbf{1}$ and $V$ are symmetric comultiplicative objects in $\mathbb{S}^\mathrm{op}$ and $\mathcal{S}$, respectively.
        \item $\mathbf{1}$ and $V$ are symmetric comonoids in $\mathbb{F}^\mathrm{op}$ and $\mathcal{F}$, respectively.
    \end{enumerate}
\end{proposition}

\subsection{Context extension}
\label{ContextExtension}

Major motivation for the presheaf approach to modelling variable binding is
that one may model the operation of context extension by an endofunctor on
$\mathcal{C}$.  

\medskip

Firstly, observe that each object $A$ in $\mathbb{C}$ induces the endofunctor
$(-) \otimes A : \mathbb{C} \to \mathbb{C}$.

\begin{definition}
For each object $A \in \mathbb{C}$, the functor of \emph{context extension} by
$A$, $\delta_A: \mathcal{C} \to \mathcal{C}$, is the presheaf restriction
along $(-) \otimes A$, given explicitly as $\delta_A(X) = X(- \otimes A)$.
\end{definition}

The presheaf $\delta_A(X)$, when evaluated at a context $B \in \mathbb{C}$,
returns the terms for $X$ in the \textit{extended context} $B \otimes A$.
Each $\delta_A$ has a left and right adjoint, induced as the left and right
\textit{Yoneda extension} of $- \otimes A$ as in the following diagram.

\[\begin{tikzcd}[ampersand replacement=\&,row sep=scriptsize]
	{\mathbb{C}^\mathrm{op}} \&\& {\mathcal{C}} \\
	\\
	{\mathbb{C}^\mathrm{op}} \&\& {\mathcal{C}}
	\arrow["{\mathcal{Y}}", hook, from=1-1, to=1-3]
	\arrow["{(- \otimes A)^\mathrm{op}}"', from=1-1, to=3-1]
	\arrow[""{name=0, anchor=center, inner sep=0}, 
  curve={height=-18pt}, from=1-3, to=3-3]
	\arrow[""{name=1, anchor=center, inner sep=0}, "{- \hat\otimes \mathcal{Y}(A)}"', curve={height=18pt}, from=1-3, to=3-3]
	\arrow["{\mathcal{Y}}"', hook, from=3-1, to=3-3]
	\arrow[""{name=2, anchor=center, inner sep=0}, "{\delta_A}"{description}, from=3-3, to=1-3]
	\arrow["\dashv"{anchor=center}, draw=none, from=1, to=2]
	\arrow["\dashv"{anchor=center}, draw=none, from=2, to=0]
\end{tikzcd}\]
%

The fact that the left adjoint to $\delta_A$ is 
$- \hat{\otimes} \mathcal{Y}(A)$ indicates that there is a canonical
isomorphism $\delta_A \cong {\mathcal{Y}(A) \multimap (-)}$.  
%
As $\delta_A$ is both a left and a right adjoint, it is monoidal with respect
to both the cartesian and cocartesian structures on $\mathcal{C}$.  However,
we emphasise that it is not, in general, monoidal with respect to the Day
convolution.

\section{Structural Endofunctors}
\label{Section:StructuralEndofunctors}

\subsection{Symmetric endofunctors}
\label{SymmetricEndofunctors}

The endofunctors, $(-) \otimes A$ on $\mathbb{C}$ and $\delta_A$ on
$\mathcal{C}$, inherit structural properties of the monoidal structure
$\otimes$ on $\mathbb{C}$.  In this section, we define what it means for an
endofunctor to be endowed with such structure and exhibit how it is induced. 

\begin{definition}
Let $(\mathrm{End}(\mathcal{E}), \circ, \mathrm{id_\mathcal{E}})$ be the
monoidal category of endofunctors on a category $\mathcal{E}$, where $\circ$
denotes functor composition.  
\begin{enumerate}
\item 
  A \emph{symmetric endofunctor} on $\mathcal{E}$ is a symmetric object in
  $\mathrm{End}(\mathcal{E})$.  

\item 
  A \emph{symmetric pointed endofunctor} on $\mathcal{E}$ is a symmetric
  pointed object in $\mathrm{End}(\mathcal{E})$.

\item 
  A \emph{symmetric multiplicative endofunctor} on $\mathcal{E}$ is a
  symmetric multiplicative object in $\mathrm{End}(\mathcal{E})$.  

\item A
  \emph{symmetric monad} on $\mathcal{E}$ is a symmetric monoid in
  $\mathrm{End}(\mathcal{E})$.
\end{enumerate}
\end{definition}

Noting that the definition of a symmetric endofunctor is self-dual, we have
dual definitions of a \textit{symmetric copointed endofunctor},
\textit{symmetric comultiplicative endofunctor}, and 
\textit{symmetric comonad}.

\begin{proposition}\label{prop_deltaSymEndo}
\begin{enumerate}
\item 
  If $\mathbb{C}$ is symmetric monoidal, every object $A$ is canonically a
  symmetric object $(A, s_A)$. 
  Then, $\big((-) \otimes A, \mathrm{id}\otimes s_A\big)$ and 
  $(\delta_A, \mathsf{swap}_A)$ are symmetric endofunctors.

\item 
  If $\mathbb{C}$ is semicocartesian monoidal, every object $A$ is canonically
  a symmetric pointed object $(A, w_A, s_A)$. 
  Then, 
  $\big((-) \otimes A, \mathrm{id}\otimes w_A, \mathrm{id}\otimes s_A\big)$
  and $(\delta_A, \mathsf{up}_A, \mathsf{swap}_A)$ are symmetric pointed
  endofunctors.

\item 
  If $\mathbb{C}$ is corelevant monoidal, every object $A$ is canonically a
  symmetric multiplicative object $(A, c_A, s_A)$. 
  Then, 
  $\big((-) \otimes A, \mathrm{id} \otimes c_A, \mathrm{id} \otimes s_A\big)$
  and $(\delta_A, \mathsf{cont}_A, \mathsf{swap}_A)$ are symmetric
  multiplicative endofunctors.

\item 
  If $\mathbb{C}$ is cocartesian monoidal, every object $A$ is canonically a
  symmetric monoid $(A, c_A, w_A, s_A)$. 
  Then, 
  $\big( 
    (-) \otimes A , 
    \mathrm{id}\otimes c_A , 
    \mathrm{id} \otimes w_A ,
    \mathrm{id} \otimes s_A 
   \big)$
  and $(\delta_A, \mathsf{cont}_A, \mathsf{up}_A, \mathsf{swap}_A)$ are
  symmetric monads.
\end{enumerate}
Above,\\[-6mm]
\begin{align*}
\mathsf{swap}_{A, X, B} 
&= X(\mathrm{id}_B \otimes s_A) : \delta_A^2(X)(B) \to \delta_A^2(X)(B) 
\\
\mathsf{up}_{A, X, B} 
&= X(\mathrm{id}_B \otimes w_A) : X(B) \to \delta_A(X)(B) 
\\
\mathsf{cont}_{A, X, B} 
&= X(\mathrm{id}_B \otimes c_A) : \delta_A^2(X)(B) \to \delta_A(X)(B)
\end{align*}
\end{proposition}

We remark that, for each $X$ in $\mathcal{C}$, the endofunctor 
$(-) \hat{\otimes} X$ is endowed with the appropriate dual structure induced
by the monoidal structure on $\mathbb{C}$, due to
Proposition~\ref{prop_daySymMon}.

\subsection{Tensorial strength}

In~\cite{K70}, Kock defines a \textit{strong endofunctor} $T$ to be one with a
\textit{(right) tensorial strength} 
$\mathbf{str}_{A,B} : T(A) \otimes B \to T(A \otimes B)$ and a 
\textit{strong monad} $(T, \mu, \eta)$ to be one such such that diagrams
\textbf{(b)} and \textbf{(c)} of Fig.~\ref{fig_strProp} commute, which
assert that the strength respects the monad structure of $T$.
This definition may be extended to  structural endofunctors by asking that the
strength also respects the symmetry.

\begin{figure*}
    \[
    \begin{tikzcd}[ampersand replacement=\&,column sep=small, row sep=scriptsize]
	{T^2(A)\!\otimes\!B} \& {T(T(A)\!\otimes\!B)} \& {T^2(A\!\otimes\!B)} \\
	{T^2(A)\!\otimes\!B} \& {T(T(A)\!\otimes\!B)} \& {T^2(A\!\otimes\!B)}
        \arrow["{\mathbf{(a)}}"{description}, draw=none, from=1-2, to=2-2]
	\arrow["{\mathbf{str}_{T(A),B}}", shift left, draw=none, from=1-1, to=1-2]
	\arrow[from=1-1, to=1-2]
	\arrow["{\varsigma_A \otimes \mathrm{id}}"', from=1-1, to=2-1]
	\arrow["{T(\mathbf{str}_{A,B})}", shift left, draw=none, from=1-2, to=1-3]
	\arrow[from=1-2, to=1-3]
	\arrow["{\varsigma_{A\otimes B}}", from=1-3, to=2-3]
	\arrow["{\mathbf{str}_{T(A),B}}"', shift right, draw=none, from=2-1, to=2-2]
	\arrow[from=2-1, to=2-2]
	\arrow["{T(\mathbf{str}_{A,B})}"', shift right, draw=none, from=2-2, to=2-3]
	\arrow[from=2-2, to=2-3]
\end{tikzcd} \;
\begin{tikzcd}[ampersand replacement=\&, cramped, row sep=scriptsize, column sep=small]
	{A\!\otimes\! B} \\
	{T(A)\!\otimes\!B} \& {T(A\!\otimes\!B)}
	\arrow["{\mathbf{(b)}}"{description, pos=0.2}, shift right=3, draw=none, from=1-1, to=2-2]
	\arrow["{\eta_A \otimes \mathrm{id}}"', from=1-1, to=2-1]
	\arrow["{\eta_{A\otimes B}}", from=1-1, to=2-2]
	\arrow["{\mathbf{str}_{A,B}}"', from=2-1, to=2-2]
\end{tikzcd} \;
    \begin{tikzcd}[ampersand replacement=\&,cramped, row sep=scriptsize,column sep=small]
	{T^2(A)\!\otimes\! B} \& {T(T(A)\!\otimes\!B)} \& {T^2(A\!\otimes\! B)} \\
	{T(A)\!\otimes\! B} \&\& {T(A\! \otimes\! B)}
        \arrow["{\mathbf{(c)}}"{description}, shift left=5, draw=none, from=2-1, to=2-3]
	\arrow["{\mathbf{str}_{T(A), B}}", shift left, draw=none, from=1-1, to=1-2]
	\arrow[from=1-1, to=1-2]
	\arrow["{\mu_A\otimes \mathrm{id}}"', from=1-1, to=2-1]
	\arrow["{T(\mathbf{str}_{A,B})}", shift left, draw=none, from=1-2, to=1-3]
	\arrow[from=1-2, to=1-3]
	\arrow["{\mu_{A\otimes B}}", from=1-3, to=2-3]
	\arrow["{\mathbf{str}_{A,B}}"', from=2-1, to=2-3]
\end{tikzcd}
    \]
    \caption{Strength properties}
    \label{fig_strProp}
\end{figure*}
\begin{figure*}
    \begin{equation*}
            \begin{tikzcd}[ampersand replacement=\&,cramped, row sep=scriptsize]
	{T \circ T \circ S} \& {T \circ S \circ T} \& {S \circ T \circ T} \\
	{T \circ T\circ S} \& {T\circ S \circ T} \& {S\circ T \circ T}
        \arrow["{\mathbf{(a)}}"{description}, draw=none, from=1-2, to=2-2]
	\arrow["{T\tau}", from=1-1, to=1-2]
	\arrow["{\varsigma S}"', from=1-1, to=2-1]
	\arrow["{\tau T}", from=1-2, to=1-3]
	\arrow["{S\tau}", from=1-3, to=2-3]
	\arrow["{T\tau}"', from=2-1, to=2-2]
	\arrow["{\tau T}"', from=2-2, to=2-3]
\end{tikzcd}
 \quad\quad
\begin{tikzcd}[ampersand replacement=\&,cramped, row sep=scriptsize]
	S \\
	{T\circ S} \& {S\circ T}
	\arrow["{\mathbf{(b)}}"{description}, shift right=3, draw=none, from=1-1, to=2-2]
	\arrow["{{\eta S}}"', from=1-1, to=2-1]
	\arrow["{{S\eta}}", from=1-1, to=2-2]
	\arrow["\tau"', from=2-1, to=2-2]
\end{tikzcd} \quad\quad
    \begin{tikzcd}[ampersand replacement=\&,cramped, row sep=scriptsize]
	{T\circ T\circ S} \& {T\circ S\circ T} \& {S\circ T \circ T} \\
	{T\circ S} \&\& {S\circ T}
        \arrow["{\mathbf{(c)}}"{description}, shift left=5, draw=none, from=2-1, to=2-3]
	\arrow["{T\tau}", from=1-1, to=1-2]
	\arrow["{\mu S}"', from=1-1, to=2-1]
	\arrow["{\tau T}", from=1-2, to=1-3]
	\arrow["{S\mu}", from=1-3, to=2-3]
	\arrow["\tau"', from=2-1, to=2-3]
\end{tikzcd}
    \end{equation*}
    \caption{Distributive law properties}
    \label{fig_distLawProp}
\end{figure*}

\begin{definition}
Let $(T,\mathbf{str})$ be a strong endofunctor.
\begin{enumerate}
\item 
  For a symmetric endofunctor $(T, \varsigma)$, $T$ is a 
  \emph{strong symmetric endofunctor} if \textbf{(a)} in 
  Fig.~\ref{fig_strProp} commutes.

\item 
  For a symmetric pointed endofunctor $(T, \eta, \varsigma)$, $T$ is a
  \emph{strong symmetric pointed endofunctor} if \textbf{(a)} and
  \textbf{(b)} in Fig.~\ref{fig_strProp} commute.

\item 
  For a symmetric multiplicative endofunctor $(T, \mu, \varsigma)$, $T$ is a
  \emph{strong symmetric multiplicative endofunctor} if \textbf{(a)} and
  \textbf{(c)} in Fig.~\ref{fig_strProp} commute.  

\item 
  For a symmetric monad $(T, \mu, \eta, \varsigma)$, $T$ is a
  \emph{strong symmetric monad} if \textbf{(a)}--\textbf{(c)} in 
  Fig.~\ref{fig_strProp} commute.
\end{enumerate}
\end{definition}

The isomorphism $\delta_A\!\cong\!{\mathcal{Y}(A)\!\multimap\!(-)}$ endows
$\delta_A$ with a \textit{left tensorial strength},
$\mathbf{str'}_{X,Y} 
 : X \hat\otimes \delta_A(Y) \to \delta_A(X \hat\otimes Y)$. 
If $\mathbb{C}$ is symmetric monoidal, then the Day convolution on
$\mathcal{C}$ is too, in which case $\delta_A$ has a right tensorial strength,
\[
\mathbf{str}
: 
\delta_A(X) \hat{\otimes} Y 
  \cong 
Y \hat\otimes \delta_A(X) 
  \xrightarrow{\mathbf{str'}} 
\delta_A(Y \hat\otimes X) 
  \cong 
\delta_A(X \hat{\otimes} Y) 
\]
which, for each type of monoidal structure on $\mathbb{C}$, respects the
induced structure on $\delta_A$.

\subsection{Symmetric distributive laws}

We perform a similar analysis on the notion of a \textit{distributive law}, by
first considering the usual notion on a monad, extending the definition to
preserve symmetry and then considering each condition individually. We
therefore obtain the following definitions.

\begin{definition}
Let $T$ and $S$ be endofunctors on a category and let 
$\tau: T\circ S \to S\circ T$ be a natural transformation.  
\begin{enumerate}
\item 
  If $(T, \varsigma)$ is a symmetric endofunctor, then $\tau$ is a
  \emph{symmetric transformation} if \textbf{(a)} in 
  Fig.~\ref{fig_distLawProp} commutes.  

\item 
  If $(T, \eta, \varsigma)$ is a symmetric pointed endofunctor, $\tau$ is a
  \emph{symmetric pointed transformation} if \textbf{(a)} and \textbf{(b)}
  in Fig.~\ref{fig_distLawProp} commute.

\item 
  If $(T, \mu, \varsigma)$ is a symmetric multiplicative endofunctor, $\tau$
  is a \emph{symmetric multiplicative transformation} if \textbf{(a)} and
  \textbf{(c)} in Fig.~\ref{fig_distLawProp} commute.  

\item 
  If $(T, \mu, \eta, \varsigma)$ is a symmetric monad, $\tau$ is a
  \emph{symmetric distributive law} if \textbf{(a)}--\textbf{(c)} in 
  Fig.~\ref{fig_distLawProp} commute.
\end{enumerate}
\end{definition}

Similar dual definitions may be given for a 
\textit{symmetric copointed transformation}, a 
\textit{symmetric comulitplicative transformation}, and a
\textit{symmetric codistributive law}. 
In fact, we have already encountered examples of such structures: The
symmetry, $\varsigma: T^2 \to T^2$, of a symmetric endofunctor, symmetric
pointed endofunctor, symmetric multiplicative endofunctor or symmetric monad
is respectively a symmetric transformation, symmetric pointed transformation,
symmetric multiplicative transformation, or symmetric distributive law between
$T$ and itself.  Strength is also an example of such a natural transformation.

\begin{proposition}\label{prop_strDistLaw}
If $F$ is a strong endofunctor on a symmetric
(\textit{resp.}~semicocartesian/corelevant/cocartesian) monoidal category
$(\otimes, I)$, then for each object $B$, the component of the strength, 
\[
  \mathbf{str}_{-,B}: (- \otimes B) \circ F \to F \circ (- \otimes B)
\]
is a symmetric transformation (\textit{resp.}~symmetric pointed
transformation/symmetric multiplicative transformation/symmetric distributive
law) between the symmetric endofunctor (\textit{resp.} symmetric pointed
endofunctor/symmetric multiplicative endofunctor/symmetric monad) 
$(-) \otimes B$ and $F$.  
\end{proposition}

In Section~\ref{BindingSignaturesAbstractSyntax}, we associate an endofunctor
to each binding signature, which is constructed using the Day tensor,
coproduct, and context extension.  In particular, it will be useful to lift
symmetric distributive laws and transformations from these structures to
signature endofunctors, and we introduce the following lemma to provide for
this.

\begin{lemma}\label{lem_lift}
Let $T$ be a symmetric endofunctor (resp.~symmetric pointed
endofunctor/symmetric multiplicative endofunctor/symmetric monad) which is
oplax monoidal on a monoidal category $(\otimes, I)$, and let $G_1$ and $G_2$
be two endofunctors on $\mathbb{C}$. 
If $\psi_1: TG_1 \to G_1T$ and $\psi : TG_2 \to G_2T$ are symmetric
transformations (resp.~symmetric pointed transformations/symmetric
multiplicative transformations/symmetric distributive laws), then the
composite 
$\Tilde{\psi} 
 : T\!\circ\!(G_1\!\otimes\!G_2) \to (G_1\!\otimes\!G_2)\!\circ\!T$ 
defined as
\[\begin{tikzcd}[ampersand replacement=\&,cramped,column sep=scriptsize,row sep=tiny]
	{T\!\circ\! (G_1\! \otimes\! G_2)\!} \& {\!(T\!\circ\! G_1) \!\otimes\! (T\!\circ\! G_2)\!} \& {\!(G_1\!\circ\! T) \!\otimes\! (G_2\!\circ\! T)}
	\arrow["{l_{G_1, G_2}}", shift left, draw=none, from=1-1, to=1-2]
	\arrow[from=1-1, to=1-2]
	\arrow["{\psi_1 \otimes \psi_2}", shift left, draw=none, from=1-2, to=1-3]
	\arrow[from=1-2, to=1-3]
\end{tikzcd}\]
is a symmetric transformation (resp.~symmetric pointed
transformation/symmetric multiplicative transformation/symmetric distributive
law).  
\end{lemma}

Note that in the lemma above, if $l$, $\psi_1$, and $\psi_2$ are isomorphisms,
then so is $\Tilde{\psi}$.

\section{Categorical Abstract Syntax}

The modelling of abstract syntax and its semantics by means of endofunctor
algebras is by now well established; see, for instance,~\cite{Aczel}.  In this
context, endofunctors model syntactic constructors while algebras equip
objects with semantic interpretation.  The abstract syntax is then an initial
(or free) algebra.  It is abstract in that it is characterised up to
isomorphism and thus representation independent.  Furthermore, by definition,
it comes equipped with iterators (referred to as catamorphisms in functional
programming~\cite{MeijerFokkingaPatersonHughesFunctionalProgrammingWithBananas}
and eliminators in type theory~\cite{ProgrammingWithMartinLofTypeTheory}) that
provide definitions by structural recursion equipped with induction proof
principles.
This section briefly reviews this theory; first in the classical sense
(Section~\ref{InitialAlgebraApproach}), followed by an important
generalisation due to Bird and Paterson \cite{BirdPaterson}
(Section~\ref{GeneralisedRecursion}), and then considering the particular case
of abstract syntax with variable binding
(Section~\ref{BindingSignaturesAbstractSyntax}). 

\subsection{Initial-algebra approach}\label{InitialAlgebraApproach}

A $\Sigma$-algebra, for an endofunctor $\Sigma$ on a category $\mathcal C$, is
a pair $(A, \alpha)$ where $A$ is an object in $\mathcal C$ and 
$\alpha: \Sigma(A) \to A$ is a morphism in $\mathcal{C}$.  
A $\Sigma$-homomorphism $h: (A, \alpha) \to (A', \alpha')$ is a morphism 
$h : A \to A'$ in $\mathcal{C}$ such that 
$h \, \alpha = \alpha' \, \Sigma(h)$.  $\Sigma$-algebras and
$\Sigma$-homomorphisms organise themselves into a category
$\mathbf{\Sigma}\textbf{-Alg}$ equipped with a forgetful functor 
$U : \mathbf{\Sigma}\textbf{-Alg} \to \mathcal{C}$, defined by 
$U(A, \alpha) = A$. 
If for every object $X$ in $\mathcal{C}$, 
$[\eta_X, \varphi_X] : X + \Sigma(TX) \to TX$ is an initial
$(X+\Sigma)$-algebra, then the forgetful functor $U$ has a left adjoint 
$F : \mathcal{C} \to \mathbf{\Sigma}\textbf{-Alg}$ defined by 
$F(X) = (TX, \varphi_X)$. 

The universal property of $F(X)$ provides a categorical model for structural
recursion.  We give it explicitly here to draw a parallel with the
generalisation of this notion presented in the forthcoming section.  For every
$\Sigma$-algebra $\alpha: {\Sigma(A) \to A}$ and morphism $\beta: X \to A$,
there exists a unique morphism $\mathsf{it}(\beta,\alpha): TX \to A$ such that
$\mathsf{it}(\beta,\alpha) \ \eta_X = \beta$ (the \emph{base case}) and
$\mathsf{it}(\beta,\alpha) \ \varphi_X 
 = 
 \alpha \ \Sigma(\mathsf{it}(\beta,\alpha))$ 
(the \emph{recursion}).
In this model, since by Lambek's Lemma~\cite{L68} initial endofunctor algebras
are isomorphisms, the base case and the recursion determine the
\textit{iterator} $\mathsf{it}(\beta,\alpha)$.

\subsection{Generalised recursion}\label{GeneralisedRecursion}

We recall a result of Bird and Paterson \cite{BirdPaterson} that generalises
the above model of structural recursion.  Its importance is that it provides
generalised iterators for free constructions over initial algebras

\begin{lemma}[Bird and Paterson~\cite{BirdPaterson}]\label{lem_BirdPaterson}
Consider an adjunction $F \dashv G : \mathcal{C'} \to \mathcal{C}$ and an
initial $S$-algebra $\alpha: {S(A) \xrightarrow\cong A}$ of an endofunctor $S$
on $\mathcal C$. 
Then, for all $S$-algebras $\gamma : SG(B) \to G(B)$ there exists a unique
\textit{generalised iterator} $\mathsf{git}(\gamma) : F(A) \to B$ in
$\mathcal{C'}$ such that the diagram on the right (and, equivalently, the
diagram on the left) below commutes
\[\begin{tikzcd}[ampersand replacement=\&,cramped]
	{S(A)} \& {SG(B)} \& {FS(A)} \& {FSG(B)} \\
	A \& {G(B)} \& {F(A)} \& B
	\arrow["{S(\overline{\mathsf{git}(\gamma)})}", from=1-1, to=1-2]
	\arrow["\alpha"', from=1-1, to=2-1]
	\arrow["\cong", draw=none, from=1-1, to=2-1]
	\arrow["\gamma", from=1-2, to=2-2]
	\arrow["{FS(\overline{\mathsf{git}(\gamma)})}", from=1-3, to=1-4]
	\arrow["{F(\alpha)}"', from=1-3, to=2-3]
	\arrow["\cong", draw=none, from=1-3, to=2-3]
	\arrow["{\overline\gamma}", from=1-4, to=2-4]
	\arrow["{\overline{\mathsf{git}(\gamma)}}", from=2-1, to=2-2]
	\arrow["\mathsf{git}(\gamma)", from=2-3, to=2-4]
\end{tikzcd}\]
where $\overline{\mathsf{git}(\gamma)}$ and $\overline{\gamma}$ are the
respective transposes of $\mathsf{git}(\gamma)$ and $\gamma$ over 
$F \dashv G$.  
\end{lemma}

We will make extensive use of the following instance of the above lemma where
the $S$-algebra $\gamma$ factors through an algebra of an endofunctor on
$\mathcal{C'}$; see, for instance,~\cite{MatthesUustalu}.

\begin{corollary}\label{Corollary:GPSR}
Under the hypothesis of Lemma~\ref{lem_BirdPaterson}, for an endofunctor $S'$
on $\mathcal C'$, a natural transformation $\psi : {FS \to S'F}$, and an 
$S'$-algebra $\beta: S'(B) \to B$, there exists a unique 
\emph{generalised iterator} $\mathsf{git}(\beta): F(A) \to B$ such that the
following commutes 
\[\begin{tikzcd}[ampersand replacement=\&,cramped, row
    sep=scriptsize]
	{S'F(A)} \&\& {S'(B )} \\
	{FS(A)} \\
	{F(A)} \&\& B
	\arrow["{S'(\mathsf{git}(\beta))}", from=1-1, to=1-3]
	\arrow["\beta", from=1-3, to=3-3]
	\arrow["{\psi_A}", from=2-1, to=1-1]
	\arrow["{F(\alpha)}"', from=2-1, to=3-1]
	\arrow["\cong", draw=none, from=2-1, to=3-1]
	\arrow["\mathsf{git}(\beta)", dashed, from=3-1, to=3-3]
\end{tikzcd}\]
Moreover, if $\psi$ is invertible, then 
$F(\alpha)\,\psi_A^{-1}: {S'F(A) \to F(A)}$ is an initial $S'$-algebra.
\end{corollary}

We remark that in the applications of this corollary below, we will consider
$S = X + \Sigma$, for an object $X$ and an endofunctor $\Sigma$, and use the
universal property of the free \mbox{$\Sigma$-algebra} on $X$.

\subsection{Abstract syntax with variable binding}
\label{BindingSignaturesAbstractSyntax}

We recall the notion of \textit{binding} (or \textit{second-order})
\textit{signature} in~\cite{FPT}.  Such signatures generalise algebraic
signatures to account for variable-binding operators.  Details on second-order
algebraic theories may be found
in~\cite{FioreSecondOrder,FioreHur,FioreMahmoud}.

\begin{definition}
A \emph{binding signature} is a pair $(\Omega, a)$ where $\Omega$ is a
\emph{set of operators} and $a : \Omega \to \mathbb{N}^*$ is an
\emph{arity function}, where $\mathbb{N}^*$ is the set of finite tuples of
natural numbers.  
\end{definition}

For an operator $\omega \in \Omega$, with arity 
$a(\omega) = (n_1, \ldots, n_k)$, $k$ is the usual arity in the algebraic
sense, specifying the number of arguments for $\omega$.  Each $n_i$
corresponds to the $i$th argument, specifying the number of variables bound by
$\omega$ in that argument.  For a category of contexts $\mathbb C$, we
associate an endofunctor on $\mathcal{C} = \mathbf{Set}^{\mathbb C}$ to each
binding signature.  This is done by first taking the coproduct over the
operators of the signature, then the Day convolution over the arguments for
each operator, and finally applying context extension $n_i$ times as
specified by the arity of the operator.  The formal definition follows.

\begin{definition}
The \emph{binding-signature endofunctor} on $\mathcal{C}$ associated to a
binding signature $(\Omega, a)$ is defined, with respect to a fixed object $A$
in $\mathbb{C}$, as
\begin{align*}\textstyle
\Sigma_{(\Omega,a)}(X) 
= \coprod_{\omega \in \Omega} \Sigma_\omega(X) 
\enspace, \quad
\Sigma_\omega(X) 
= \widehat{\bigotimes}_{i \in \mathbf{k}} \delta^{n_i}_A(X)
\end{align*}
\end{definition}

\begin{figure*}
\begin{align*}
\mathbf{swap}_X 
= \delta\left(\coprod_{\omega \in \Omega}\prod_{i \in \mathbf{k}} \delta^{n_i}(X) \right) \xrightarrow{\cong} \coprod_{\omega \in \Omega}\prod_{i \in \mathbf{k}} \delta \delta^{n_i}(X) \xrightarrow{\coprod\prod \mathsf{swap}^{n_i}} \coprod_{\omega \in \Omega}\prod_{i \in \mathbf{k}} \delta^{n_i}\delta(X) 
\\
\mathbf{str}_{X, Y} 
= \left( \coprod_{\omega \in \Omega}\prod_{i \in \mathbf{k}} \delta^{n_i}(X) \right) \times Y \to \coprod_{\omega \in \Omega}\prod_{i \in \mathbf{k}} \delta^{n_i}(X) \times Y \xrightarrow{\mathbf{str}^{n_i}} \coprod_{\omega \in \Omega}\prod_{i \in \mathbf{k}} \delta^{n_i}(X \times Y) 
\end{align*}
\caption{Swap and strength natural transformations for 
  $\Sigma: \mathcal{F} \to \mathcal{F}$} 
\label{fig_swapStrSigF}
\end{figure*}
\begin{figure*}
    \[\begin{tikzcd}[ampersand replacement=\&,cramped,column sep=small, row sep=scriptsize ]
	{J \hat\otimes X} \& X \\
	{\Sigma_\mathsf{sub}(X)}
        \arrow["{\mathbf{(a)}}"{description, pos=0.25}, shift left=3, draw=none, from=2-1, to=1-2]
	\arrow["\cong", from=1-1, to=1-2]
	\arrow["{\nu\hat\otimes \mathrm{id}}"', from=1-1, to=2-1]
	\arrow["\sigma"', from=2-1, to=1-2]
\end{tikzcd}
\quad
\begin{tikzcd}[ampersand replacement=\&,cramped,column sep=small, row sep=scriptsize]
	{\delta(X) \hat\otimes J} \&\& {\delta(X)} \\
	{\delta(X) \hat\otimes \delta(X)} \& {\delta\Sigma_\mathsf{sub}(X)}
	\arrow["\cong", from=1-1, to=1-3]
	\arrow[""{name=0, anchor=center, inner sep=0}, "{{\mathrm{id} \hat\otimes\nu}}"', from=1-1, to=2-1]
	\arrow["{{\mathbf{str'}}}", from=2-1, to=2-2]
	\arrow[""{name=1, anchor=center, inner sep=0}, "{{\delta(\sigma)}}"', from=2-2, to=1-3]
	\arrow["{{\mathbf{(b)}}}"{description}, shift left, draw=none, from=0, to=1]
\end{tikzcd}
\quad
\begin{tikzcd}[ampersand replacement=\&,cramped,column sep=small, row sep=scriptsize]
{\delta(X)\hat\otimes\delta(X)\hat\otimes X} \&\& {\delta\Sigma_\mathsf{sub}(X) \hat\otimes X} \\
	{\Sigma_\mathsf{sub}(X)} \& X \& {\Sigma_\mathsf{sub}(X)}
	\arrow["{{\mathbf{str'}\hat\otimes \mathrm{id}}}", from=1-1, to=1-3]
	\arrow[""{name=0, anchor=center, inner sep=0}, "{{\mathrm{id}\hat\otimes \sigma}}"', from=1-1, to=2-1]
	\arrow[""{name=1, anchor=center, inner sep=0}, "{{\delta(\sigma)\hat\otimes \mathrm{id}}}", from=1-3, to=2-3]
	\arrow["\sigma", from=2-1, to=2-2]
	\arrow["\sigma"', from=2-3, to=2-2]
	\arrow["{\mathbf{(c)}}"{description}, shift left, draw=none, from=0, to=1]
\end{tikzcd}\]
\[\begin{tikzcd}[ampersand replacement=\&,cramped, row sep=scriptsize]
	{\delta^2(X) \hat\otimes X \hat\otimes X} \&\& {\delta^2(X) \hat\otimes X \hat\otimes X} \&\& {\delta^2(X) \hat\otimes X \hat\otimes X} \\
	{\delta\Sigma_\mathsf{sub}(X)\hat\otimes X} \& {\Sigma_\mathsf{sub}(X)} \& X \& {\Sigma_\mathsf{sub}(X)} \& {\delta\Sigma_\mathsf{sub}(X)\hat\otimes X}
	\arrow["{\mathbf{(d)}}"{description}, draw=none, from=1-3, to=2-3]
	\arrow["{\mathrm{swap}\hat\otimes\mathrm{id}}", from=1-1, to=1-3]
	\arrow["{\mathbf{str}\hat\otimes \mathrm{id}}"', from=1-1, to=2-1]
	\arrow["{\mathrm{id}\hat\otimes \cong}", from=1-3, to=1-5]
	\arrow["{\mathbf{str}\hat\otimes \mathrm{id}}", from=1-5, to=2-5]
	\arrow["{\delta(\sigma) \hat\otimes \mathrm{id}}", from=2-1, to=2-2]
	\arrow["\sigma", from=2-2, to=2-3]
	\arrow["\sigma"', from=2-4, to=2-3]
	\arrow["{\delta(\sigma) \hat\otimes \mathrm{id}}"', from=2-5, to=2-4]
\end{tikzcd}
\quad
\begin{tikzcd}[ampersand replacement=\&,cramped, row sep=scriptsize]
	{X \hat\otimes X} \& X \\
	{\Sigma_\mathsf{sub}(X)}
        \arrow["{\mathbf{(e)}}"{description, pos=0.25}, shift left=3, draw=none, from=2-1, to=1-2]
	\arrow["{\pi_1}", from=1-1, to=1-2]
	\arrow["{\mathsf{up}_X\hat\otimes\mathrm{id}}"', from=1-1, to=2-1]
	\arrow["\sigma"', from=2-1, to=1-2]
\end{tikzcd}
\]
\[\begin{tikzcd}[ampersand replacement=\&,cramped,column sep=large, row sep=scriptsize]
	{\delta\delta(X) \hat\otimes \delta(X) \hat\otimes X} \& {\Sigma_\mathsf{sub}\delta(X)\hat\otimes X} \& {\Sigma_\mathsf{sub}\Sigma_\mathsf{sub}(X)} \& {\Sigma_\mathsf{sub}(X)} \\
	{\delta\Sigma_\mathsf{sub}(X) \hat\otimes X} \& {\Sigma_\mathsf{sub}(X)} \&\& X
	\arrow["{{\mathsf{swap}\hat\otimes\mathrm{id}\hat\otimes \mathrm{id}}}", from=1-1, to=1-2]
	\arrow[""{name=0, anchor=center, inner sep=0}, "{{\rho \hat\otimes \mathrm{id}}}"', from=1-1, to=2-1]
	\arrow["{{\mathbf{str}^\mathsf{sub}}}", from=1-2, to=1-3]
	\arrow["{{\Sigma_\mathsf{sub}(\sigma)}}", from=1-3, to=1-4]
	\arrow[""{name=1, anchor=center, inner sep=0}, "\sigma", from=1-4, to=2-4]
	\arrow["{{\delta(\sigma)\hat\otimes \mathrm{id}}}", from=2-1, to=2-2]
	\arrow["\sigma", from=2-2, to=2-4]
	\arrow["{\mathbf{(f)}}"{description}, draw=none, from=0, to=1]
\end{tikzcd}\]
    \caption{Substitution algebra axioms}
    \label{fig_SubstAlgAx}
\end{figure*}

While one may define a binding signature endofunctor for any $A$ in 
$\mathbb C$, in the cases of interest ($\mathcal{F}$, $\mathcal{B}$,
$\mathcal{I}$, $\mathcal S$) the category of contexts 
($\mathbb F$, $\mathbb B$, $\mathbb I$, $\mathbb S$) are freely generated by
the object $\mathbf{1}$ and we naturally restrict attention to it.  
We will therefore only consider the operation of context extension
$\delta_{\mathbf 1}$, writing it simply as $\delta$, and the signature
endofunctor induced by it.
The \emph{abstract syntax of a binding signature} $(\Omega,a)$ arises then as
the free $\Sigma_{(\Omega,a)}$-algebra on the presheaf of variables $V$
(equivalently, the initial 
$(V+\Sigma_{(\Omega,a)})$-algebra).

\section{Cartesian Theory}
\label{Section:CartesianTheory}

We revisit the theory of single-variable (capture-avoiding) substitution for
abstract syntax with variable binding of Fiore, Plotkin and Turi~\cite{FPT}.
We fully exploit the categorical theory thus far developed to provide new
constructions
and streamlined direct proofs. 

\subsection{Cartesian substitution algebras}
\label{CartesianSubstitutionAlgebras}

We begin by further studying signature endofunctors defined on $\mathcal{F}$.
By Proposition~\ref{prop_daySymMon}, the Day convolution coincides with the
cartesian product, so $J = 1$, the terminal presheaf.  In particular, signature
endofunctors are defined using the cartesian monoidal structure.

We iteratively apply Lemma~\ref{lem_lift} to the symmetric distributive law
$\mathsf{swap}: \delta^2 \to \delta^2$ to obtain a symmetric distributive law
$\mathbf{swap}: \delta\Sigma \xrightarrow{\cong} \Sigma\delta$.  This uses the
fact that $\delta$ is monoidal with respect to both the cartesian and
cocartesian tensors and is explicitly given in Fig.~\ref{fig_swapStrSigF}.

Recalling that for each object $Y$ in $\mathcal{F}$, the diagonal morphism
$\Delta_Y: Y \to Y \times Y$ makes the symmetric comonad $- \times Y$ on
$\mathcal{F}$ oplax monoidal with respect to the cartesian tensor, we
similarly apply Lemma~\ref{lem_lift} to the symmetric codistributive law
$\mathbf{str}_{-, Y} : (- \times Y)\circ\delta \to \delta\circ(- \times Y)$ of
Proposition~\ref{prop_strDistLaw} to obtain, for a binding-signature
endofunctor $\Sigma$, a symmetric codistrubutive law
$\mathbf{str}_{-, Y} : (- \times Y)\circ\Sigma \to \Sigma \circ (- \times Y)$
explicitly given in Fig.~\ref{fig_swapStrSigF}. Thus, binding-signature
endofunctors are strong.

We now describe an algebraic structure in $\mathcal{F}$ that axiomatises
single-variable substitution for cartesian theories.  Such a definition first
appeared in~\cite{FPT} and the equivalent variation below featured
in~\cite{FioreStaton}.  However, aiming at a unified theory for substructural
syntax, the definition below provides a presentation making use of the
categorical structures thus far developed.  Specifically, we consider the
\emph{substitution signature} $\Sigma_\mathsf{sub} = \delta(-) \times (-)$,
and note that it is equipped with a strength $\mathbf{str}^\mathsf{sub}$ and
swapping isomorphism $\mathbf{swap}^\mathsf{sub}$.

\begin{definition}\label{defn_cartSubstAlg}
A \emph{cartesian substitution algebra} is a triple $(X, \sigma, \nu)$ where
$X$ in an object in $\mathcal{F}$, and $\sigma: \Sigma_\mathsf{sub}(X) \to X$
and $\nu : 1 \to \delta (X)$ are morphisms in $\mathcal{F}$ such that
\textbf{(a)}, \textbf{(b)}, \textbf{(e)} and \textbf{(f)} in 
Fig.~\ref{fig_SubstAlgAx} commute, where $\hat\otimes = \times$, $J = 1$, and
$\rho$ is the isomorphism exhibiting $\delta$ as monoidal.  
\end{definition}

In this definition, $\sigma$ is the operation of substitution for $X$, while
$\nu$ specifies the generic variables for $X$.  Each of the axioms are
understood as follows: 
\textbf{(a)} is a left-unit law and says that substituting a term into a
variable returns the term;
\textbf{(b)} is a right-unit law and says that substituting in a variable
amounts to performing a contraction;
\textbf{(e)} says that substituting into a weakened term does nothing; 
while \textbf{(f)} is the \textit{syntactic substitution lemma}, which
specifies an associativity law for substitutions.

To clarify the understanding of \textbf{(b)}, we remark that in the context of
this definition, it may be equivalently replaced by the following one
\[\begin{tikzcd}[ampersand replacement=\&,cramped, row sep=scriptsize]
	{\delta^2(X) \times 1} \&\& {\delta^2(X)} \\
	{\delta^2(X)  \times \delta(X)} \& {\delta\Sigma_\mathsf{sub}(X)} \& {\delta(X)}
	\arrow["{\pi_1}", "\cong"', from=1-1, to=1-3]
	\arrow["{\mathrm{id} \times \nu}"', from=1-1, to=2-1]
	\arrow["{\mathsf{cont}_X}", from=1-3, to=2-3]
	\arrow["\cong", from=2-1, to=2-2]
	\arrow["{\delta(\sigma)}", from=2-2, to=2-3]
\end{tikzcd}\]

A morphism $f : (X, \sigma, \nu) \to (X', \sigma', \nu')$ of cartesian
substitution algebras is a morphism ${f: X \to X'}$ in $\mathcal{F}$ such that
$\delta(f) \ \nu = \nu'$ and 
$f \ \sigma = \sigma' \ \Sigma_\mathsf{sub}(f)$.

Cartesian substitution algebras and their morphisms organise into a category
$\mathbf{CSubstAlg}$.  We note the following result, a proof of which appears
in~\cite{TACPaper}, that justifies the axiomatisation.

\begin{theorem}[Fiore, Plotkin and Turi~\cite{FPT}]
\label{thm_CSubstAlgClnLaw}
The category of cartesian substitution algebras, the category of abstract
clones, the category of Lawvere theories, and the category of cartesian
one-object multicategories are equivalent.
\end{theorem}

\subsection{Cartesian abstract syntax}
\label{CartesianAbstractSyntax}

Recall from Section~\ref{subsec_UniDisc} the presheaf of variables, 
$V = \mathcal{Y}(\mathbf{1})$, which by Proposition~\ref{prop_VSymOb} is a
symmetric comonoid. This presheaf models cartesian variables as an object in
$\mathcal{F}$.  The abstract syntax of a binding signature is modelled by the
free $\Sigma$-algebra over $V$.  We denote this by 
$\varphi_V: \Sigma(\TV) \to \TV$ together with the morphism 
$\eta_V: V \to \TV$ provided by the initial $(V + \Sigma)$-algebra structure.
The following results appear in~\cite{FPT}.  However, we note that new direct
categorical proofs are available using Corollary~\ref{Corollary:GPSR}.

\begin{lemma}[Fiore, Plotkin and Turi~\cite{FPT}]\label{lem_TVCSubstAlg}
$\TV$ is equipped with a canonical cartesian substitution algebra structure.  
\end{lemma}

\begin{proposition}[Fiore, Plotkin and Turi~\cite{FPT}]
The $\Sigma$-algebra 
\[\begin{tikzcd}[ampersand replacement=\&,cramped,row sep=small]
	{\Sigma\delta(\TV)} \& {\delta\Sigma(\TV)} \& {\delta(\TV)}
	\arrow["{\mathbf{swap}^{-1}}", from=1-1, to=1-2]
	\arrow["\cong"', draw=none, from=1-1, to=1-2]
	\arrow["{\delta(\varphi_V)}", from=1-2, to=1-3]
\end{tikzcd}\]
together with the morphism $\delta(\eta_V): \delta (V)\to \delta (TV)$ present
$\delta (TV)$ as a free $\Sigma$-algebra over $\delta (V) \cong V+1$.
\end{proposition}

We direct our attention to the universal property of the substitution algebra
$(\TV, \sigma, \nu)$; namely, that $\varphi_V$ is the initial $\Sigma$-algebra
with compatible substitution-algebra structure. To express this fact, we
recall the following definition.

\begin{definition}[Fiore, Plotkin and Turi~\cite{FPT}]
A \textit{cartesian $\Sigma$-substitution algebra} is a quadruple 
$(X, \sigma, \nu, \alpha)$ where $(X, \sigma, \nu)$ is a cartesian
substitution algebra and $(X, \alpha)$ is a $\Sigma$-algebra such that the
following diagram commutes
\[\begin{tikzcd}[ampersand replacement=\&,cramped, row sep=scriptsize]
	{\Sigma\Sigma_\mathsf{sub}(X)} \&\& {\Sigma(X)} \\
	{\delta\Sigma(X) \times X} \\
	{\Sigma_\mathsf{sub}(X)} \&\& X
	\arrow["{\Sigma(\sigma)}", from=1-1, to=1-3]
	\arrow["\alpha", from=1-3, to=3-3]
	\arrow["{\mathbf{str}(\mathbf{swap}\times \mathrm{id})}", from=2-1, to=1-1]
	\arrow["{\delta(\alpha)\times \mathrm{id}}"', from=2-1, to=3-1]
	\arrow["\sigma", from=3-1, to=3-3]
\end{tikzcd}\]
\end{definition}

A morphism of such structures is a morphism in $\mathcal{F}$ that is both a
$\Sigma$-homomorphism and a morphism of cartesian substitution algebras, and
we obtain the category $\Sigma$-$\mathbf{CSubstAlg}$. 

\begin{theorem}[Fiore, Plotkin and Turi~\cite{FPT}]
For a signature $\Sigma$, $(\TV, \sigma, \nu, \varphi_V)$ is an initial object
in $\Sigma$-$\mathbf{CSubstAlg}$.  
\end{theorem}
\begin{proof}[Proof (idea)]
Lemma~\ref{lem_TVCSubstAlg} indicates that $(\TV, \sigma, \nu, \varphi_V)$ is
an object in \mbox{$\Sigma$-$\mathbf{CSubstAlg}$}.  Regarding it being
initial, the unique 
morphism to any other cartesian \mbox{$\Sigma$-substitution} algebra is
induced by the initial $(V + \Sigma)$-algebra $[\varphi_V, \eta_V]$.  The fact
that it is a morphism in $\Sigma$-$\mathbf{CSubstAlg}$ follows by an
application of Corollary~\ref{Corollary:GPSR}.  
\end{proof}

\section{Linear Theory}
\label{Section:LinearTheory}

We now consider single-variable substitution for linear theories, 
left open by Tanaka~\cite{Tanaka} when developing the case of simultaneous
substitution.
As mentioned at the end of the introduction, this involves the crucial
development of \textit{derived functors} for signature endofunctors
(Section~\ref{subsec_linDerFun}), which are needed to account for the specific
interaction between context extension and the term pairing that occurs in the
linear setting.

\subsection{Linear substitution algebras}

Recall from Proposition~\ref{prop_ctxSymMon} that the category of contexts for
linear theories, $\mathbb{B}$, is symmetric monoidal.  The Day convolution on
the universe of discourse, $\mathcal{B}$, is also symmetric monoidal
(Proposition~\ref{prop_daySymMon}) and $\delta: \mathcal{B} \to \mathcal{B}$
is a symmetric endofunctor (Proposition~\ref{prop_deltaSymEndo}).  
The following axiomatisation of single-variable substitution (referred to as
partial composition in operad theory~\cite{Markl}) for \emph{non-unital}
linear theories first appeared in~\cite{FioreCT2014}.  However, again aiming
at a unified theory for substructural syntax, the definition below provides a
presentation making use of the developed categorical structures.

\begin{definition}\label{defn_linSubstAlg}
A \textit{linear substitution algebra} is a triple $(X, \sigma, \nu)$ where
$X$ is an object in $\mathcal{B}$, and $\sigma: \Sigma_\mathsf{sub}(X) \to X$
and $\nu: J \to \delta(X)$ are morphisms in $\mathcal{B}$ such that
\textbf{(a)}, \textbf{(b)}, \textbf{(c)}, and \textbf{(d)} in 
Fig.~\ref{fig_SubstAlgAx} commute.  
\end{definition}


\begin{figure*}
\[\begin{tikzcd}[ampersand replacement=\&,cramped,column sep=large, row sep=scriptsize]
	{\delta\Sigma_\mathsf{sub}(X)\hat\otimes X} \& {\Sigma_\mathsf{sub}^\dagger(X, \delta(X))\hat\otimes X} \& {\Sigma_\mathsf{sub}^\dagger(X, \Sigma_\mathsf{sub}(X))} \& {\Sigma_\mathsf{sub}^\dagger(X)} \\
	{\Sigma_\mathsf{sub}(X)} \&\&\& X
	\arrow["{\mathbf{swap}^\mathsf{sub}\hat\otimes \mathrm{id}}", from=1-1, to=1-2]
	\arrow["{\delta(\sigma) \hat\otimes \mathrm{id}}"', from=1-1, to=2-1]
	\arrow["{\mathbf{str}^\mathsf{sub}}", from=1-2, to=1-3]
	\arrow["{\Sigma_\mathsf{sub}^\dagger(\mathrm{id}, \sigma)}", from=1-3, to=1-4]
	\arrow["{\sigma^\dagger}", from=1-4, to=2-4]
	\arrow["\sigma", from=2-1, to=2-4]
\end{tikzcd}\]
    \caption{Extended substitution lemma}
    \label{fig_ExtSubLem}
\end{figure*}

\begin{figure*}
    {\small \begin{center} \begin{tabular}{|c|c|c|}
    \hline
         & $\Sigma_\mathsf{sub}^\dagger(X,Y)$ & $\mathbf{swap}^\mathsf{sub}$ and $\mathbf{str}^\mathsf{sub}$ \\
         \hline
        $\mathcal{B}$ & $\delta(Y) \hat\otimes X\!+\!\delta(X) \hat\otimes Y$ & \tlc{$\mathbf{swap}^\mathsf{sub} = \delta(\delta(X)\hat\otimes X) \xrightarrow{(\mathsf{swap}\hat\otimes \mathrm{id} + \mathrm{id})\mathcal{L}} \delta^2(X)\hat\otimes X + \delta(X)\hat\otimes\delta(X)$}{$\mathbf{str}^\mathsf{sub} = (\delta(Y)\hat{\otimes}X + \delta(X) \hat\otimes Y) \hat\otimes Z \to \delta(Y\hat\otimes Z) \hat{\otimes}X + \delta(X) \hat\otimes Y \hat\otimes Z$} \\
        \hline
        $\mathcal{I}$ & $\!\delta(Y) \hat\otimes X\!+\!\delta(X) \hat\otimes Y\!+\!\delta(X) \hat\otimes X$ & \tlc{$\mathbf{swap}^\mathsf{sub} = \delta(\delta(X)\hat\otimes X)  \xrightarrow{(\mathbf{swap}\hat\otimes \mathrm{id}\!+\!\mathrm{id})\mathcal{K}} \delta^2(X)\hat\otimes X \!+\! \delta(X) \hat\otimes \delta(X) \!+\! \delta(X) \hat\otimes X$}{$\mathbf{str}^\mathsf{sub} = (\delta(Y) \hat\otimes X \!+\! \delta(X) \hat\otimes Y \!+\! \delta(X) \hat\otimes X) \hat\otimes Z \to \delta(Y \hat\otimes Z) \hat\otimes X  \!+\! \delta(X) \hat\otimes Y \hat\otimes Z \!+\! \delta(X) \hat\otimes X$} \\
        \hline
        $\mathcal{S}$ & $\delta(Y) \hat\otimes X \!+\! \delta(X) \hat\otimes Y \!+\! \delta(Y) \hat\otimes Y$ & \tlc{$\mathbf{swap}^\mathsf{sub} = \delta(\delta(X) \hat\otimes X) \xrightarrow{(\mathbf{swap}\hat\otimes \mathrm{id} + \mathrm{id} + \mathbf{swap} \hat\otimes \mathrm{id}) \mathcal{H}} \delta^2(X) \hat\otimes X \!+\! \delta(X) \hat\otimes \delta(X) \!+\! \delta^2(X) \hat\otimes \delta(X)$}{$\mathbf{str}^\mathsf{sub} = (\delta(Y) \hat\otimes X \!+\! \delta(X) \hat\otimes Y \!+\! \delta(Y) \hat\otimes Y) \hat\otimes Z \to \delta(Y\hat\otimes Z) \hat\otimes X \!+\! \delta(X) \hat\otimes Y \hat\otimes Z \!+\! \delta(Y \hat\otimes Z) \hat\otimes Y \hat\otimes Z$} \\
        \hline
    \end{tabular} \end{center} }
    \caption{Derived functor with $\mathbf{swap}$ and $\mathbf{str}$ morphisms
      for $\Sigma_\mathsf{sub}(-) = \delta(-) \hat\otimes (-)$}
    \label{tab_derSubSigEndo}
\end{figure*} 

As in Definition~\ref{defn_cartSubstAlg}, $\sigma$ is the operation of
substitution and $\nu$ specifies generic variables, and axioms \textbf{(a)}
and \textbf{(b)} have the same intuitive understanding.  Axioms \textbf{(c)}
and \textbf{(d)} are the operad (or multicategory) laws of associativity and
exchange~\cite{Markl}, which model the behaviour of sequential and parallel
composition, respectively.  Comparing this to
Definition~\ref{defn_cartSubstAlg}, observe the absence of axiom \textbf{(e)}:
there is no weakening on linear contexts for substitution to respect.  An
\textit{extended substitution lemma} is in fact encoded in the associativity
and exchange for operads, and in Section~\ref{subsec_linDerFun} we develop a
categorical construction of a \textit{linear derived functor} for a signature
endofunctor to express this.  Furthermore, this is also required in the theory
of linear abstract syntax of Section~\ref{Subsection:LinearAbstractSyntax}.

The morphisms of linear substitution algebras are similar to those of
cartesian substitution algebras, and these structures organise themselves into
a category, $\mathbf{LSubstAlg}$.  We have the following result, analogous to
Theorem~\ref{thm_CSubstAlgClnLaw}.

\begin{theorem}
The category of linear substitution algebras, the category of symmetric
operads, and the category of one-object symmetric multicategories are
equivalent.  
\end{theorem}

\subsection{Linear derived functors}
\label{subsec_linDerFun}

In contrast to cartesian theories, the Day convolution in $\mathcal{B}$ does
not coincide with the cartesian monoidal tensor, and $\delta$ is not monoidal
with respect to it.  In particular, for a signature endofunctor $\Sigma$, one
may not apply Lemma~\ref{lem_lift} to induce swapping and strength on
$\Sigma$.  Instead, observe that the universal morphism
$[\mathbf{str}, \mathbf{str'}]
 : \delta(X) \hat\otimes Y + X \hat\otimes \delta(Y) 
   \to 
   \delta(X \hat\otimes Y)$ 
is an isomorphism, say with inverse 
\begin{equation*}
 \mathcal{L} 
 : \delta(X \hat\otimes Y) 
   \xrightarrow{\cong} 
   \delta(X) \hat\otimes Y + X \hat\otimes \delta(Y)
\end{equation*}

$\mathcal{L}$ is referred to as the \textit{Leibniz isomorphism} and makes
$\delta$ a derivative operator on $\mathcal{B}$~\cite{J81, J86}.  As in
classical differential calculus, the Leibniz isomorphism may be recursively
applied to a finitary $\hat\otimes$-product, 
$\widehat\bigotimes_{i\in \mathbf{n}} X_i$, to obtain a natural isomorphism
denoted by $\mathcal{L}^n$.

Consider a signature endofunctor 
$\Sigma = \coprod_{\omega \in \Omega} \Sigma_\omega$.  Each $\Sigma_\omega$
(with $a(\omega) = (n_1, \ldots, n_k)$) is a $k$-ary $\hat\otimes$-product, so
one may apply $\mathcal{L}^k$ to $\delta\Sigma_\omega$, followed by the
isomorphism $\mathsf{swap}^{n_i}: \delta\delta^{n_i} \to \delta^{n_i}\delta$
to the newly introduced $\delta$ in each term. To express this, we define the
\textit{linear derived functor} for an operator $\omega$.

\begin{definition}
For an operator $\omega \in \Omega$, the \emph{linear derived functor} of
$\Sigma_\omega$ is the bifunctor 
$\Sigma_\omega^\dagger: \mathcal{B} \times \mathcal{B} \to \mathcal{B}$ with
$\Sigma_\omega^\dagger(X, Y)$ given by
\begin{equation*}\textstyle
\coprod_{j \in \mathbf{k}} 
  \left(\!\! 
    \left( \widehat{\bigotimes_{i \in \mathbf{j-1}}} \delta^{n_i} (X) \!\right)
    \!\hat{\otimes} 
    \delta^{n_j}(Y) \hat\otimes\!
      \left(
        \widehat{\bigotimes_{i \in \mathbf{k-j}}} \delta^{n_{j+i}}(X) 
      \right)
    \!\!\right) 
\end{equation*}
\end{definition}

This functor comes canonically equipped with a swapping isomorphism:
\[
\mathbf{swap}_{\omega,X} 
=  
\delta\Sigma_\omega(X) 
  \xrightarrow
    {\left(
       \coprod_{j \in \mathbf{k}}
         \mathrm{id} \hat\otimes \mathsf{swap}^{n_j} \hat\otimes \mathrm{id}
     \right)\mathcal{L}^k} 
\Sigma_\omega^\dagger(X, \delta(X)) 
\]
Additionally, because there is only one instance of $Y$ in each summand of 
$\Sigma_\omega^\dagger(X, Y)$, the bifunctor admits a strength in its second
argument: 
\[
\mathbf{str}_{\omega,X, Y, Z} 
 : \Sigma_\omega^\dagger(X , Y) \hat\otimes Z 
   \to 
   \Sigma_\omega^\dagger (X , Y \hat\otimes Z)
\]

To complete the construction, we define the \textit{linear derived functor} of
a signature $\Sigma$ as 
$\Sigma^\dagger = \coprod_{\omega \in \Omega}\Sigma_\omega^\dagger$.

Noting that $\delta$ is monoidal with respect to the cocartesian monoidal
tensor, we define a swapping isomorphism for the signature endofunctor:
\begin{gather*}\textstyle
\mathbf{swap}_X \!
\colon 
  \delta\Sigma(X) 
    \xrightarrow{\cong} 
  \coprod_{\omega \in \Omega} \delta\Sigma_\omega(X) 
    \xrightarrow{\coprod\mathbf{swap}_{\omega,X}}  
  \Sigma^\dagger(X, \delta(X)) 
\end{gather*}
Moreover, due to the Day convolution distributing over the cocartesian
monoidal tensor, $\Sigma^\dagger$ admits a strength in its second argument.
\begin{gather*}\textstyle
\mathbf{str}_{X,Y,Z}\!
\colon \!
  \Sigma^\dagger\!(\!X, Y) \hat\otimes Z \!
    \xrightarrow{\cong}\!\!
  \coprod_{\omega \in \Omega}\!\!\Sigma_\omega^\dagger\!(\!X, Y)\hat\otimes Z \!
    \xrightarrow{\!\coprod \!\mathbf{str}_\omega\!}\!\! 
  \Sigma^\dagger\!(\!X, \! Y \hat\otimes Z)
\end{gather*}
As an illustrative example, and for use in the next proposition, we consider
the construction for the substitution signature endofunctor,
$\Sigma_\mathsf{sub} = \delta(-) \hat\otimes (-)$ in
Fig.~\ref{tab_derSubSigEndo}.

We also note that the linear derived functor induces a 
\textit{linear derived endofunctor}, obtained by evaluating at the diagonal.
This endofunctor is simply a coproduct of the operators of the signature:
\begin{equation*}\textstyle
\Sigma^\dagger(X) 
= \Sigma^\dagger(X, X) 
= \coprod_{\omega \in \Omega} \coprod_{i \in \mathbf{k}} \Sigma_\omega(X)
\end{equation*}
In particular, for a $\Sigma$-algebra $\alpha: \Sigma(X) \to X$, we write 
$\alpha^\dagger: \Sigma^\dagger(X) \to X$ for the $\Sigma^\dagger$-algebra
induced by the components of $\alpha$.

We conclude the section by using these linear derived functors to provide an
equivalent definition of a linear substitution algebra in terms of an extended
substitution lemma.

\begin{proposition}\label{prop_linSubAlgAlt}
A triple $(X, \sigma , \nu)$ is a linear substitution algebra if and only if
it satisfies \textbf{(a)} and \textbf{(b)} in Fig.~\ref{fig_SubstAlgAx} and
the diagram in Fig.~\ref{fig_ExtSubLem}.  
\end{proposition}

\subsection{Linear abstract syntax}
\label{Subsection:LinearAbstractSyntax}

Recall from Proposition~\ref{prop_VSymOb} that the presheaf of variables, 
$V = \mathcal{Y}(\mathbf{1})$, is a symmetric object in $\mathcal{B}$. As
before, the abstract syntax for a binding signature is modelled by the free
\mbox{$\Sigma$-algebra} on the presheaf of variables $V$ and is denoted by 
$\varphi_V: \Sigma(\TV) \to \TV$, with $\eta_V : V \to \TV$.

\begin{lemma}\label{lem_TVLSubstAlg}
$\TV$ is equipped with a canonical linear substitution algebra structure.
\end{lemma}

\begin{proposition}
The $\Sigma^\dagger(\TV,-)$-algebra,
\[\begin{tikzcd}[ampersand replacement=\&,cramped]
	{\Sigma^\dagger(\TV, \delta(\TV))} \& {\delta\Sigma(\TV)} \& {\delta(\TV)}
	\arrow["{\mathbf{swap}^{-1}}", from=1-1, to=1-2]
	\arrow["\cong"', draw=none, from=1-1, to=1-2]
	\arrow["{\delta(\varphi_V)}", from=1-2, to=1-3]
\end{tikzcd}\]
together with the morphism $\delta(\eta_V): \delta(V) \to \delta(\TV)$
present $\delta(\TV)$ as a free $\Sigma^\dagger(\TV,-)$-algebra over
$\delta(V) \cong J$.
\end{proposition}

\begin{definition}
A \emph{linear $\Sigma$-substitution algebra} is a quadruple 
$(X, \sigma, \nu, \alpha)$ where $(X, \sigma, \nu)$ is a linear substitution
algebra and $(X, \alpha)$ is a $\Sigma$-algebra such that the following
diagram commutes
\begin{equation}\label{diag_SubstAlgMorCond}
\begin{tikzcd}[ampersand replacement=\&,cramped, row sep=scriptsize]
	{\Sigma^\dagger(X,\Sigma_\mathsf{sub}(X))} \&\& {\Sigma^\dagger(X)} \\
	{\delta\Sigma(X)\hat\otimes X} \\
	{\Sigma_\mathsf{sub}(X)} \&\& X
	\arrow["{\Sigma^\dag(\mathrm{id}, \sigma)}", from=1-1, to=1-3]
	\arrow["{\alpha^\dagger}", from=1-3, to=3-3]
	\arrow["{\mathbf{str}(\mathbf{swap} \hat\otimes\mathrm{id})}", from=2-1, to=1-1]
	\arrow["{\delta(\alpha)\hat\otimes \mathrm{id}}"', from=2-1, to=3-1]
	\arrow["\sigma", from=3-1, to=3-3]
\end{tikzcd}
\end{equation}
\end{definition}

A morphism for such structures is a morphism in $\mathcal{B}$ that is both a
$\Sigma$-algebra homomorphism and a morphism of linear substitution algebras,
and we obtain the category \mbox{$\Sigma$-$\mathbf{LSubstAlg}$}.

\begin{theorem}
For a binding signature $\Sigma$, $(\TV, \sigma, \nu, \varphi_V)$ is an
initial object in $\Sigma$-$\mathbf{LSubstAlg}$.  
\end{theorem}
\begin{proof}[Proof (idea)]
Lemma~\ref{lem_TVLSubstAlg} indicates that $(\TV, \sigma, \nu, \varphi_V)$ is
an object in \mbox{$\Sigma$-$\mathbf{LSubstAlg}$}.  The unique morphism to any
other linear $\Sigma$-substitution algebra is induced by the initial 
\mbox{$(V + \Sigma)$-algebra} $[\varphi_V, \eta_V]$. The fact that it is a
morphism in $\Sigma$-$\mathbf{LSubstAlg}$ follows by an application of
Corollary~\ref{Corollary:GPSR}.  
\end{proof}

\section{Affine Theory}
\label{Section:AffineTheory}

We consider the affine theory proceeding analogously to that of the previous
linear case.  Further details on the universe of discourse $\mathcal{I}$ may
be found in~\cite{FioreMoggiSangiorgiLICS,FioreMoggiSangiorgi}.

\subsection{Affine substitution algebras}

Recall from Proposition~\ref{prop_daySymMon} that the Day convolution in
$\mathcal{I}$ is semicartesian monoidal, so $J = 1$, and thus the tensor is
equipped with projections. By Proposition~\ref{prop_deltaSymEndo}, 
$\delta: \mathcal{I} \to \mathcal{I}$ is a symmetric pointed endofunctor.

\begin{definition}\label{defn_affSubAlg}
An \emph{affine substitution algebra} is a triple $(X, \sigma, \nu)$ where
$X$ is an object in $\mathcal{I}$, and $\sigma: \Sigma_\mathsf{sub}(X) \to X$
and $\nu: 1 \to \delta(X)$ are morphisms in $\mathcal{I}$ such that
\textbf{(a)}, \textbf{(b)}, \textbf{(c)}, \textbf{(d)} and \textbf{(e)} in
Fig.~\ref{fig_SubstAlgAx} commute.  
\end{definition}

In this definition, the axioms \textbf{(a)} and \textbf{(b)} are interpreted
as those of Definition~\ref{defn_cartSubstAlg}, while \textbf{(c)} and
\textbf{(d)} are the operad laws of Definition~\ref{defn_linSubstAlg}.  Axiom
\textbf{(e)} is the third axiom of Definition~\ref{defn_cartSubstAlg}.  In
Section~\ref{subsec_affDerFun}, we develop the notion of an 
\textit{affine derived functor} and express the axioms \textbf{(c)},
\textbf{(d)}, and \textbf{(e)} as an \emph{extended substitution lemma} that
embodies the associativity laws of affine single-variable substitution.

The morphisms of affine substitution algebras are similar to those of
cartesian and linear substitution algebras, and we obtain the category
$\mathbf{ASubstAlg}$.

In~\cite{TanakaPower}, Tanaka and Power develop a substitution tensor product
in $\mathcal{I}$, similar to those of~\cite{FPT} in $\mathcal{F}$
and~\cite{Kelly05} in $\mathcal{B}$, for which the monoids model simultaneous
substitution.

\begin{theorem}
The category of affine substitution algebras and the category of monoids for
the substitution tensor in $\mathcal{I}$ are equivalent.  
\end{theorem}

\subsection{Affine derived functors}\label{subsec_affDerFun}

As was the case for linear theories, $\delta: \mathcal{I} \to \mathcal{I}$ is
not monoidal with respect to the Day convolution in $\mathcal{I}$. 
Observe that the morphism 
\[
[\mathbf{str}, \mathbf{str'}, \mathsf{up}]
: \delta(X) \hat\otimes Y + X \hat\otimes \delta(Y) + X\hat\otimes Y 
   \to
   \delta(X \hat\otimes Y)
\]
is an isomorphism, say with inverse
\begin{equation*}
\mathcal{K}
: \delta(X \hat\otimes Y) 
  \xrightarrow{\cong} 
  \delta(X) \hat\otimes Y + X \hat\otimes \delta(Y) + X\hat\otimes Y
\end{equation*}
Recursively applying $\mathcal{K}$ to a finitary $\hat\otimes$-product,
$\widehat\bigotimes_{i \in \mathbf{n}}X_i$, yields a natural isomorphism
$\mathcal{K}^n$.

As in Section~\ref{subsec_linDerFun}, we will consider a signature endofunctor
$\Sigma = \coprod_{\omega \in \Omega}\Sigma_\omega$ and apply $\mathcal{K}^k$
to $\delta\Sigma_\omega$, for each $\omega$, followed by $\mathsf{swap}^{n_i}$
to the newly introduced $\delta$.  To this end, we define
the \textit{affine derived functor} for $\Sigma_\omega$.

\begin{definition}
For an operator $\omega \in \Omega$, the \emph{affine derived functor} of 
$\Sigma_\omega$ is the bifunctor  
$\Sigma^\dagger_\omega : \mathcal{I} \times \mathcal{I} \to \mathcal{I}$
with $\Sigma^\dagger_\omega (X, Y)$ given by 
\begin{align*}\small
\coprod_{j \in \mathbf{k}}\!\!
  \left( \!\!
    \left( 
      \widehat{\bigotimes_{i \in \mathbf{j-1}}}\!\delta^{n_i}\!(X) 
    \!\!\right)
    \!\!\hat{\otimes}
    \delta^{n_j}\!(Y) 
    \hat\otimes\!\!
    \left( 
      \widehat{\bigotimes_{i \in \mathbf{k-j}}} \!\delta^{n_{j+i}}\!(X) 
    \!\!\right) 
  \!\!\right)
\!\!+\!\!\!\!
\coprod_{i \in \mathbf{k-1}} \!\!\!\Sigma_\omega(X) 
\end{align*}
\end{definition}

Observe that the first summand of the above expression is the formula for the
linear derived functor.  Thus, the affine derived functor is similarly
equipped with a swapping isomorphism and a strength:  
\begin{align*}
\mathbf{swap}_{\omega}X 
& = 
\delta\Sigma_\omega(X) 
  \xrightarrow{\mathbf{swap} \mathcal{K}^k}  
\Sigma^\dagger_\omega(X, \delta(X)) 
\\
\mathbf{str}_{X, Y, Z} 
& = 
\Sigma^\dagger_\omega(X, Y) \hat\otimes Z 
  \xrightarrow{\mathbf{str_\omega} + \pi_1} 
\Sigma^\dagger_\omega(X, Y \hat\otimes Z) 
\end{align*}

Define the \textit{affine derived functor} for a signature endofunctor as
$\Sigma^\dagger = \coprod_{\omega \in \Omega} \Sigma^\dagger_\omega$.  The
above morphisms, together with the facts that $\delta$ is monoidal with
respect to the cocartesian tensor and that the Day convolution distributes
over the cocartesian tensor, induce a swapping isomorphism and a strength:
\begin{align*}
\mathbf{swap}_X & : \delta\Sigma(X) \to \Sigma^\dagger(X, \delta(X))
\\
\mathbf{str}_{X,Y,Z}
& : \Sigma^\dagger(X, Y)\hat\otimes Z \to \Sigma^\dagger(X, Y \hat\otimes Z)
\end{align*}
In Fig.~\ref{tab_derSubSigEndo}, we consider the substitution
signature endofunctor, $\Sigma_\mathsf{sub} = \delta(-) \hat\otimes (-)$, to
illustrate the construction.

Evaluating the bifunctor $\Sigma^\dagger$ at the diagonal provides the
endofunctor on $\mathcal I$:
\begin{equation*}
\Sigma^\dagger(X) 
= \Sigma^\dagger(X, X) 
= \coprod_{\omega \in \Omega} 
    \coprod_{i \in \mathbf{k}} 
      \Sigma_\omega(X) 
      + 
      \coprod_{\omega \in \Omega} \coprod_{i \in \mathbf{k-1}} \Sigma_\omega(X)
\end{equation*}
Therefore, for each $\Sigma$-algebra, $\alpha: \Sigma(X) \to X$, we obtain a
$\Sigma^\dagger$-algebra, written 
$\alpha^\dagger : \Sigma^\dagger(X) \to (X)$, induced by the components of
$\alpha$. 

We may now state an analogous result to Proposition~\ref{prop_linSubAlgAlt}
which, using the isomorphism $\mathcal{K}$, captures the last three axioms of
Definition~\ref{defn_affSubAlg} as an extended substitution lemma.

\begin{proposition}\label{prop_affSubAlgAlt}
A triple $(X, \sigma, \nu)$ is an affine substitution algebra if and only if
it satisfies \textbf{(a)} and \textbf{(b)} in Fig.~\ref{fig_SubstAlgAx}, and
the diagram in Fig.~\ref{fig_ExtSubLem}.  
\end{proposition}

\subsection{Affine abstract syntax}

Recall from Proposition~\ref{prop_VSymOb} that the presheaf of variables 
$V = \mathcal{Y}(\mathbf{1})$ is a symmetric copointed object in
$\mathcal{I}$. The abstract syntax of a binding signature is modelled as the
free $\Sigma$-algebra on $V$, denoted by $\varphi_V: \Sigma(\TV) \to \TV$,
with $\eta_V : V \to \TV$.

\begin{lemma}\label{lem_TVASubstAlg}
$\TV$ is equipped with a canonical affine substitution algebra structure.
\end{lemma}
\begin{proof}[Proof (idea)]
The morphism $\nu : 1 \to \delta(\TV)$ is the transpose of $\eta_V$ over the
adjunction $(-) \hat\otimes V \dashv \delta$. 
One then observes that $\delta(V) \cong V + 1$ and defines 
\textit{basic substitution} as in the cartesian case:
\begin{equation*}
\beta 
\, = \, 
\delta(V) \hat\otimes \TV 
  \xrightarrow{\cong} 
V \hat\otimes \TV + 1 \hat\otimes \TV 
  \xrightarrow{[\eta_V\pi_1, \pi_2]} 
\TV
\end{equation*}
The proof is concluded by invoking Proposition~\ref{prop_affSubAlgAlt} and
appropriately defining an \textit{affine second derived functor}. 
\end{proof}

Denote the canonical affine substitution algebra on $\TV$ by 
$(\TV, \sigma, \nu)$.

\begin{proposition}
The $\Sigma^\dagger(\TV, -)$-algebra, 
\[\begin{tikzcd}[ampersand replacement=\&,cramped]
	{\Sigma^\dagger(\TV, \delta(\TV))} \& {\delta\Sigma(\TV)} \& {\delta(\TV)}
	\arrow["{\mathbf{swap}^{-1}}", from=1-1, to=1-2]
	\arrow["\cong"', draw=none, from=1-1, to=1-2]
	\arrow["{\delta(\varphi_V)}", from=1-2, to=1-3]
\end{tikzcd}\] 
together with the morphism $\delta(\eta_V) : \delta(V) \to \delta(\TV)$ 
present $\delta(\TV)$ as a free $\Sigma^\dagger(\TV,-)$-algebra over 
$\delta(V) \cong V+1$.
\end{proposition}

\begin{definition}
An \emph{affine $\Sigma$-substitution algebra} is a quadruple 
$(X, \sigma, \nu, \alpha)$ where $(X, \sigma, \nu)$ is an affine substitution
algebra and $(X, \alpha)$ is a $\Sigma$-algebra such that
(\ref{diag_SubstAlgMorCond}) commutes.
\end{definition}

A morphism for such structures is a morphism in $\mathcal{I}$ that is both a
$\Sigma$-algebra homomorphism and a morphism of affine substitution algebras.
We obtain the category $\Sigma$-$\mathbf{ASubstAlg}$.

\begin{theorem}
For a signature $\Sigma$, $(\TV, \sigma, \nu , \varphi_V)$ is an initial
object in $\Sigma$-$\mathbf{ASubstAlg}$.  
\end{theorem}
\begin{proof}[Proof (idea)]
That $(\TV, \sigma, \nu, \varphi_V)$ is 
in $\Sigma$-$\mathbf{ASubstAlg}$ follows from Lemma~\ref{lem_TVASubstAlg}. 
The unique morphism to any other affine $\Sigma$-substitution algebra is
induced by the initial \mbox{$(V + \Sigma)$-algebra} $[\varphi_V, \eta_V]$.
The fact that it is a morphism in $\Sigma$-$\mathbf{ASubstAlg}$ follows by an
application of Corollary~\ref{Corollary:GPSR}.  
\end{proof}

\section{Relevant Theory}
\label{Section:RelevantTheory}

\subsection{Relevant substitution algebras}

Recall from Proposition~\ref{prop_daySymMon} that the Day convolution on
$\mathcal{S}$ is relevant monoidal and thus it is equipped with diagonals. 
By Proposition~\ref{prop_deltaSymEndo}, $\delta: \mathcal{S} \to \mathcal{S}$
is a symmetric multiplicative endofunctor. In this case, $\delta$ is a lax
monoidal endofunctor with respect to the Day convolution, witnessed by the
natural transformation
\begin{equation*}
\rho_{X, Y} 
= 
\delta(X) \hat\otimes \delta(Y) 
  \xrightarrow{\delta(\mathbf{str'})\mathbf{str}} 
\delta^2(X \hat\otimes Y) 
  \xrightarrow{\mathsf{cont}} 
\delta(X \hat\otimes Y)
\end{equation*}

\begin{definition}\label{defn_revSubAlg}
A \emph{relevant substitution algebra} is a triple $(X, \sigma, \nu)$ where
$X$ is an object $\mathcal{S}$, and $\sigma: \Sigma_\mathsf{sub}(X) \to X$ and
$\nu : J \to \delta(X)$ are morphisms in $\mathcal{S}$ such that \textbf{(a)},
\textbf{(b)}, \textbf{(c)}, \textbf{(d)}, and \textbf{(f)} in
Fig.~\ref{fig_SubstAlgAx} commute.  
\end{definition}

In this definition, axioms \textbf{(a)} and \textbf{(b)} are the same as those
of Definition~\ref{defn_linSubstAlg} and axioms \textbf{(c)} and \textbf{(d)}
are the operad laws.  Axiom \textbf{(f)} is precisely the substitution lemma
of Definition~\ref{defn_cartSubstAlg}.  In the next section we will again
develop derived functors for this theory, providing an alternative
axiomatisation which captures axioms \textbf{(c)}, \textbf{(d)}, and
\textbf{(f)} as an extended substitution lemma.

The morphisms of relevant substitution algebras are similar to those of the
previous cases, and we obtain the category $\mathbf{RSubstAlg}$.

\subsection{Relevant derived functors}
\label{RelevantDerivedFunctors}

Observe that the morphism 
\[
[\mathbf{str}, \mathbf{str'}, \rho]
 \colon 
 \delta(X) \hat\otimes Y + X \hat\otimes \delta(Y) 
   + \delta(X) \hat\otimes \delta(Y)
 \to 
 \delta(X \hat\otimes Y)
\]
is an isomorphism, say with inverse
\begin{equation*}
\mathcal{H} 
: 
\delta(X \hat\otimes Y) 
  \xrightarrow{\cong} 
\delta(X) \hat\otimes Y + X \hat\otimes \delta(Y) 
  + \delta(X) \hat\otimes \delta(Y)
\end{equation*}

As before, we apply $\mathcal{H}$ recursively to a finitary
$\hat\otimes$-product, $\widehat{\bigotimes}_{i \in \mathbf{n}}X_i$.  The
codomain of the resulting isomorphism is the coproduct over all instances of
the finitary product, where at least one of the $X_i$ has $\delta$ applied to
it.  To express this formally, let $S(n) = \{0, 1\}^n \setminus \{0\}^n$ be 
the set of all binary $n$-tuples excluding the tuple containing only $0$s.
Note that this set is equipped with $n$ projections 
$\pi_i : S(n) \to \{0, 1\}$. Then,
\begin{equation*}\textstyle
\mathcal{H}^n 
: 
\delta\left( \widehat{\bigotimes_{i \in \mathbf{n}}} X_i \right) 
  \xrightarrow{\cong} 
\coprod_{j \in S(n)} 
  \widehat{\bigotimes_{i \in \mathbf{n}}} \delta^{\pi_i(j)}(X_i) 
\end{equation*}

Considering the signature endofunctor 
$\Sigma = \coprod_{\omega \in \Omega} \Sigma_\omega$ we aim to define a
\textit{relevant derived functor} for each $\Sigma_\omega$. For objects $X$
and $Y$ in $\mathcal{S}$, let $\ell_{X,Y}: \{0, 1 \} \to \{X, Y\}$ be the
labelling function defined by $\ell_{X,Y}(0) = X$ and $\ell_{X,Y}(1) = Y$.

\begin{definition}
For an operator $\omega \in \Omega$, the \emph{relevant derived functor} of
$\Sigma_\omega$ is the bifunctor 
$\Sigma^\dagger_\omega : \mathcal{S} \times \mathcal{S} \to \mathcal{S}$
defined by
\begin{equation*}\textstyle
\Sigma^\dagger_\omega (X, Y) 
= \coprod_{j \in S(k)}
    \widehat{\bigotimes_{i \in \mathbf{k}}}\ \delta^{n_i}(\ell_{X,Y}(\pi_i(j)))
\end{equation*}
\end{definition}

As previously, the relevant derived functor is canonically equipped with a
swapping isomorphism 
\begin{equation*}
\mathbf{swap}_X\!\colon \!\!\delta\Sigma_\omega\!(X) \!\!
\xrightarrow{\!\mathcal{H}^k\!\!} \!\!\!\!\!
\coprod_{j \in S(k)} \!
  \widehat{\bigotimes_{i \in \mathbf{k}}} 
    \delta^{\pi_i\!(j)}\!\delta^{n_i}\!(X)\!\! 
\xrightarrow{\!\coprod \!\widehat\bigotimes \mathsf{swap}^{n_i}\!\!\!} \!\!
\Sigma^\dagger_\omega\!(\!X,\! \delta(X)\!)
\end{equation*}
The Day convolution has diagonals and distributes over the cocartesian tensor,
so the strength of $\delta$ equips the relevant derived functor with a
strength in its second argument: 
\[
\mathbf{str}_{X,Y,Z} 
: \Sigma^\dagger_\omega (X, Y) \hat\otimes Z 
  \to 
  \Sigma^\dagger_\omega(X, Y \hat\otimes Z)
\]

The \textit{relevant derived functor for a signature} endofunctor is
$\Sigma^\dagger = \coprod_{\omega \in \Omega} \Sigma^\dagger$.  It is equipped
with a swapping isomorphism and a strength:
\begin{align*}
\mathbf{swap}_X & 
: \delta\Sigma(X) \xrightarrow{\cong} \Sigma^\dagger(X, \delta(X))
\\
\mathbf{str}_{X, Y, Z}  & 
: \Sigma^\dagger(X, Y) \hat\otimes Z \to \Sigma^\dagger(X, Y \hat\otimes Z)
\end{align*}
In Fig.~\ref{tab_derSubSigEndo}, we consider the substitution signature
endofunctor $\Sigma_\mathsf{sub} = \delta(-) \hat\otimes -$, to illustrate
the construction.

Evaluating the bifunctor $\Sigma^\dagger$ at the diagonal yields the
endofunctor on $\mathcal{S}$:
\begin{equation*}\textstyle
\Sigma^\dagger(X) 
= \Sigma^\dagger(X, X) 
= \coprod_{\omega\in\Omega}\coprod_{j \in S(k)} \Sigma_\omega(X)
\end{equation*}
Thus, a $\Sigma$-algebra, $\alpha: \Sigma(X) \to X$, induces a 
$\Sigma^\dagger$-algebra, $\alpha^\dagger : \Sigma^\dagger(X) \to X$.

\begin{proposition}\label{prop_relSubAlgAlt}
A triple $(X, \sigma, \nu)$ is a relevant substitution algebra if and only if
it satisfies \textbf{(a)} and \textbf{(b)} in Fig.~\ref{fig_SubstAlgAx}, and
the diagram in Fig.~\ref{fig_ExtSubLem}.  
\end{proposition}

\subsection{Relevant abstract syntax}

Recall from Proposition~\ref{prop_VSymOb} that the presheaf of variables 
$V = \mathcal{Y}(\mathbf{1})$ is a symmetric comultiplicative object in
$\mathcal{S}$. The abstract syntax of a binding signature is modelled as the
free $\Sigma$-algebra on $V$, denoted by $\varphi_V: \Sigma(\TV) \to \TV$,
with $\eta_V : V \to \TV$.

\begin{lemma}\label{lem_TVRSubstAlg}
$\TV$ is equipped with a canonical relevant substitution algebra structure
$(\sigma, \nu)$.
\end{lemma}
\begin{proof}[Proof (idea)]
By invoking Proposition~\ref{prop_relSubAlgAlt} and appropriately defining a
\textit{relevant second derived functor}. 
\end{proof}

\begin{proposition}
The $\Sigma^\dagger(\TV,-)$-algebra, 
\[\begin{tikzcd}[ampersand replacement=\&,cramped,row sep=small]
	{\Sigma^\dagger(\TV, \delta(\TV))} \& {\delta\Sigma(\TV)} \& {\delta(\TV)}
	\arrow["{\mathbf{swap}^{-1}}", from=1-1, to=1-2]
	\arrow["\cong"', draw=none, from=1-1, to=1-2]
	\arrow["{\delta(\varphi_V)}", from=1-2, to=1-3]
\end{tikzcd}\] 
together with the morphism $\delta(\eta_V) : \delta(V) \to \delta(\TV)$
present $\delta(\TV)$ as a free $\Sigma^\dagger(\TV,-)$-algebra over
$\delta(V) \cong J$.
\end{proposition}

\begin{definition}
A \textit{relevant $\Sigma$-substitution algebra} is a quadruple 
$(X, \sigma, \nu, \alpha)$ where $(X, \sigma, \nu)$ is a relevant substitution
algebra and $(X, \alpha)$ is a $\Sigma$-algebra such that
(\ref{diag_SubstAlgMorCond}) commutes.  
\end{definition}

A morphism for such structures is a morphism in $\mathcal{S}$ that is both a
$\Sigma$-homomorphism and a morphism of relevant substitution algebras. We
obtain the category $\Sigma$-$\mathbf{RSubstAlg}$.

\begin{theorem}
For a signature $\Sigma$, $(\TV, \sigma, \nu, \varphi_V)$ is an initial object
in $\Sigma$-$\mathbf{RSubstAlg}$.  
\end{theorem}
\begin{proof}
That $(\TV, \sigma, \nu, \varphi_V)$ is 
in $\Sigma$-$\mathbf{RSubstAlg}$ follows from Lemma~\ref{lem_TVRSubstAlg}.
The unique morphism to any other relevant $\Sigma$-substitution algebra is
induced by the initial $(V + \Sigma)$-algebra $[\varphi_V, \eta_V]$.  The fact
that it is a morphism in $\Sigma$-$\mathbf{RSubstAlg}$ follows by an
application of Corollary~\ref{Corollary:GPSR}.  
\end{proof}

\section*{Concluding remarks}

We have established a theory of substructural abstract syntax with variable
binding focussing on single-variable substitution.  There are scientific,
theoretical, and practical reasons for the latter.

Scientifically, the study of single-variable substitution has been somewhat
relegated in favor of that of simultaneous substitution and our work remedies
this situation.

Theoretically, the categorical machinery behind the development of
single-variable substitution is in some respects more elementary than the one
needed for simultaneous substitution.  In particular, the latter requires the
development of substitution tensor products, based on Kan extensions and/or
coends, 
which hinder formalization for computation. 

Practically, our approach is even novel in the traditional cartesian case.
Indeed, while a direct transcription of the single-variable substitution
program of
~\cite[Section~3]{FPT} is not well-typed in current dependently-typed proof
assistants, our theory may be used to mathematically derive a program that is.
Furthermore, an  application that will be presented elsewhere is the
development of a formalisation of normalisation for simply typed lambda
calculus by hereditary substitution.  This is particularly suited because it
is in this more general context that single-variable substitution is the
fundamental notion, and not a derived one.

As for speculation, a possible application, along the lines of what was done
in
~\cite{FioreStaton}, 
is the investigation of the linear, affine, and relevant algebraic theories
put forward in this paper as 
computational effects.  It is perhaps in this context that connections with
the resource theory of 
lambda calculi may arise.
Furthermore, our theory of `derived functors' may be of independent
theoretical interest.  In this direction, and in connection to theories of
differentiation, we point out that our `affine product rule' (see
Fig.~\ref{Figure:ProductRules}) has recently also been considered by
Par\'e~\cite{Pare} in the categorical study of the `difference operator'.



\end{document}